\documentclass[preprint,12pt]{elsarticle}

\usepackage{amssymb}
\usepackage{amsthm}

\newtheorem{theorem}{Theorem}
\newtheorem{corollary}[theorem]{Corollary}
\newtheorem{lemma}[theorem]{Lemma}
\newtheorem{definition}{Definition}
\newtheorem{example}{Example}
\newtheorem{remark}{Remark}

\usepackage{amsmath}
\usepackage{color}
\usepackage{bm}
\usepackage{url}
\usepackage{combelow}  % Era\cb{s}cu
\usepackage{csquotes}
\usepackage{nicefrac}
\usepackage{lhelp}

\usepackage{algorithmic,multirow}

\usepackage[ruled,vlined,linesnumbered]{algorithm2e}

\def\cont{\mathop{\rm cont}\nolimits}
\def\prim{\mathop{\rm prim}\nolimits}
\def\disc{\mathop{\rm disc}\nolimits}
\def\ldcf{\mathop{\rm ldcf}\nolimits}
\def\coeff{\mathop{\rm coeff}\nolimits}
\def\res{\mathop{\rm res}\nolimits}

\def\Res{\mathop{\rm Res}\nolimits}

\journal{}
\let\today\relax

\begin{document}

\begin{frontmatter}



\title{Cylindrical Algebraic Decomposition \\ with Equational Constraints}

%\author{Matthew England, Russell Bradford and James H. Davenport}


\author[COV]{Matthew~England}
\ead{Matthew.England@coventry.ac.uk}
\author[BATH]{Russell~Bradford}
\ead{R.J.Bradford@bath.ac.uk}
\author[BATH]{James H.~Davenport}
\ead{J.H.Davenport@bath.ac.uk}
\address[COV]{Faculty of Engineering, Environment and Computing, Coventry University, UK}
\address[BATH]{Faculty of Science, University of Bath, UK}

\begin{abstract}
Cylindrical Algebraic Decomposition (CAD) has long been one of the most important algorithms within Symbolic Computation, as a tool to perform quantifier elimination in first order logic over the reals.  More recently it is finding prominence in the Satisfiability Checking community as a tool to identify satisfying solutions of problems in nonlinear real arithmetic.

The original algorithm produces decompositions according to the signs of polynomials, when what is usually required is a decomposition according to the truth of a formula containing those polynomials.  One approach to achieve that coarser (but hopefully cheaper) decomposition is to reduce the polynomials identified in the CAD to reflect a logical structure which reduces the solution space dimension: the presence of Equational Constraints (ECs).

This paper may act as a tutorial for the use of CAD with ECs: we describe all necessary background and the current state of the art.  In particular, we present recent work on how McCallum's theory of reduced projection may be leveraged to make further savings in the lifting phase: both to the polynomials we lift with and the cells lifted over.  We give a new complexity analysis to demonstrate that the double exponent in the worst case complexity bound for CAD reduces in line with the number of ECs.  We show that the reduction can apply to both the number of polynomials produced and their degree.  
\end{abstract}


\end{frontmatter}

\section{Introduction}
\label{SEC:Intro}

\subsection{Cylindrical algebraic decomposition}
\label{SUBSEC:IntroCAD}

A Cylindrical Algebraic Decomposition (CAD) splits $\mathbb{R}^n$ into cells to maintain an invariance structure relative to an input.  Traditionally, the cells are produced to be \emph{sign-invariant} for a set of input polynomials: meaning that throughout each cell, each of those polynomials has a constant sign.  
The first CAD algorithm was introduced by \citet{Collins1975} to perform Quantifier Elimination (QE) over real closed fields.  We describe the necessary details and terminology for CAD in Section \ref{SUBSEC:CADBackground}.  The invariance property of a CAD means that problems for non-linear polynomial systems such as QE are reduced to testing a finite number of sample points; with the nature of the cells produced allowing for the easy generation of solution descriptions.  

However, the use of CAD is often limited by the complexity of computing one.  CAD is known to have worst case complexity doubly exponential in the number of variables \citep{DH88, BD07}.  Broadly speaking (see Theorem \ref{thm:Complexity1} for a precise result) if the input has $m$ polynomials of degree at most $d$ then CAD complexity could be in the order of $(2dm)^{2^{O(n)}}$.  For some problems there exist algorithms with better complexity (see the textbook by \citet{BPR06} for example), and there are also many specialised algorithms for restricted inputs; but CAD implementations remain the only general purpose approach for many problems. 

This complexity statement is obtained by considering the necessary size of the output for certain examples, and so can only be reduced with changes to the output requirements.  Further, unlike some other theoretical results, this one is clearly felt in practice: when increasing dimensionality one will hit the ``doubly exponential wall'' where progress becomes infeasible.  
However, extensive investigation into CAD and its sub-algorithms has allowed for the wall to be ``pushed back'' to the point of allowing many useful computations. 

Applications of CAD include: 
motion planning \citep{SS83II}, 
weight minimisation for truss design \citep{CC18}, 
epidemic modelling \citep{BENW06}, 
steady state analysis of biological networks \citep{BDEEGGHKRSW17},
economic reasoning \citep{MBDET18},
artificial intelligence to pass exams \citep{WMTA16},
parametric optimisation \citep{FPM05}, 
theorem proving \citep{Paulson2012}, 
derivation of optimal numerical schemes \citep{EH16},
reasoning with multi-valued functions \citep{DBEW12}, 
and much more.

\subsection{CAD and \textsf{SC$^\mathsf{2}$}}
\label{SUBSEC:IntroSC2}

CAD has long been important within Symbolic Computation with implementations in multiple computer algebra systems.  However, in recent years CAD has been of interest to the separate community of Satisfiability Checking.  There, search based algorithms developed for the Boolean SAT problem, make use of heuristics and learning (see the textbook by \citet{BHvMW09} for details).  Success here led to research on domains other than the Booleans, and Satisfiability Module Theory (SMT)-Solvers which use SAT algorithms on the Boolean skeleton of a problem with queries to theory solvers to see if a satisfying Boolean assignment is valid in the domain (learning new clauses if not) \citep{BSST09, KS13b}.  
%See the chapter by \citet{BSST09} or the textbook by \citet{KS13b} for details.

For the SMT domain of non-linear real arithmetic (NRA), CAD and more generally computer algebra systems can play the role of such theory solvers\footnote{However, as discussed by \citet{AAB+16a} a more custom approach is beneficial.}.  \textsc{SMT-RAT} contains a tailored CAD implementation for use in SMT \citep{LSCAB13, KA19}, while \textsc{Z3} contains an algorithm by \citet{JdM12} which uses CAD theory without producing actual CADs.  The latter inspired new developments in symbolic computation such as non-uniform CAD by  \citet{Brown2015}.  Further collaboration is informed by the \textsf{SC$^\mathsf{2}$} project: forging interaction between \textbf{S}ymbolic \textbf{C}omputation and \textbf{S}atisfiability \textbf{C}hecking \citep{AAB+16a}.

CADs are produced relative to a problem statement expressed in logic connectives between atoms involving (potentially non-linear) polynomials with integer coefficients.  The original CAD algorithm produces decompositions according to the signs of these polynomials, essentially ignoring the logical structure entirely and so producing decompositions fine enough to solve all problems for all logical formulae formed by those polynomials.  A Satisfiability Checking approach like that of \citet{JdM12} takes an opposite focus, analysing and extending the logical skeleton of the formula until a solution is found with the correct algebraic properties. 
In order to derive a full solution from a CAD what is truly required is a decomposition on whose cells the truth of the overall logical formula is constant\footnote{A sign-invariant decomposition for the polynomials in the formula achieves this, but with far more cells and computation than required.}.

\subsection{Equational constraints}
\label{SUBSEC:IntroEC}

A CAD complexity analysis, such as that in Section \ref{SEC:Complexity}, does not just conclude the upper bound $(2dm)^{2^{O(n)}}$: it actually shows that in one dimension we must isolate the roots of at most $M$ polynomials of degree $D$, where $D=d^{2^{O(n)}}$ and $M=m^{2^{O(n)}}$.  The same orders have been found for lower bounds: for $D$ by \citet{DH88} and for $M$ by \citet{BD07}.
%\footnote{Both \citep{DH88} and \citep{BD07} rely on the technique from \citep{Heintz1983}.}. 
The formulae demonstrating this are not straightforward but the underlying polynomials \emph{are} surprisingly simple (all bar two linear with each only involving a bounded number of variables, independent of $n$).
%: all linear for \citep{BD07} and all bar two linear for \citep{DH88}. 
%Furthermore, each polynomial only involves a bounded number of variables (generally two) independent of $n$.  
This demonstrates that the difficulty of CAD resides in the complicated number of ways simple polynomials can interact.  Improvement must come from reducing the number of interactions we track. 

\begin{definition}
\label{def:QFF}
A {\em Quantifier Free Tarski Formula (QFF)} is made up of a finite number of atoms connected by the standard Boolean operators $\land$, $\lor$ and $\neg$.  The atoms are statements about the signs of polynomials with integer coefficients: $f \, \sigma \, 0$ where $\sigma \in \{=, <, >\}$ (and by combination also $\{\leq, \geq, \neq\}$). 
\end{definition}

\begin{definition}
\label{def:EC}
An {\em Equational Constraint} (EC) is a polynomial equation logically implied by a QFF.  If it is an atom of the formula it is said to be {\em explicit} and if not then it is {\em implicit}.
\end{definition}

\begin{example}
\label{ex:EC}
Let $f$ and $g$ be polynomials.  (a) The formula $f=0 \land g>0$ has explicit EC $f=0$.  (b) The formula $f=0 \lor g=0$ has no explicit EC but it does have the implicit EC $fg=0$.  (c) The formulae $f^2 + g^2 \leq 0$ also has no explicit EC but this one has two implicit ECs $f = 0$ and $g = 0$.
\end{example}

\citet{Collins1998a} was the first to suggest that CAD could be simplified in the presence of an EC.  He noted that a CAD need only be sign-invariant for the defining polynomial of an EC, and sign-invariant for any others only within those cells where the EC polynomial is zero.  
He sketched an intuitive approach to produce this by refining the polynomials identified by his CAD algorithm.  This approach was formalised and verified by \citet{McCallum1999b}.  
% (we summarise in Section \ref{SUBSEC:McC2}).
A complexity analysis \citep[Section 2]{BDEMW16} showed that making use of a single EC in this way reduces the double exponent of $m$ in the complexity bound for CAD by $1$.  Some natural questions arising are: 
\begin{itemizeshort}
\item Can savings be made iteratively in the presence of multiple ECs?
\item Do those savings further reduce the double exponent?
\item Can corresponding savings be made for the double exponent of $d$?
\end{itemizeshort}
The first question was answered affirmatively by \citet{McCallum2001}, although the extension was not trivial\footnote{and contained a small mistake as described in Remark \ref{rem:McC01a}.}.  The other questions are answered affirmatively by the present paper.
Such questions are of growing importance as CAD finds new application domains with increasing numbers of equations: e.g. in biology by \citet{BDEEGGHKRSW17} and \citet{EEGRSW17}; and in economics by \citet{MBDET18, MDE18}.  Indeed, many problems that arise in the Satisfiability Checking context contain far more equalities than inequalities (see the NRA benchmarks in the SMT-LIB \citep{SMTLIB}).

\subsection{Contribution and plan}

In Section \ref{SEC:Background} we define the necessary CAD terminology and revise the theory for projection, and reduced projection in the presence of an EC, of \citet{McCallum1998, McCallum1999b, McCallum2001}.  Then in Section \ref{SEC:Lifting} we present recent work on how to leverage this for savings elsewhere in CAD.  In Section \ref{SEC:Alg} we propose and verify the corresponding algorithm.  

We demonstrate the benefit first with a worked example in Section \ref{SEC:Example} and then a complexity analysis in Section \ref{SEC:Complexity}.  The latter observes the double exponent in the bound on the number of projection polynomials reducing by the number of ECs.  In Section \ref{SEC:DegreeTheory} and \ref{SEC:DegreeEvaluation} we explain how a similar reduction can be observed for their degree if we assume CAD input is pre-processed with a Gr\"{o}bner Basis: a common step in CAD implementations but this is the first theoretical justification for it.  Together these show that CAD is doubly exponential in number of variables minus number of ECs.

In Section \ref{SEC:Primitive} we examine the main caveat: an assumption of primitivity to the ECs.  We demonstrate its presence in the key results of \citet{DH88, BD07} proving worst case CAD complexity, and discuss what might be done.  We finish by discussing lessons for the Satisfiability Checking community who may call CAD within SMT-solvers in Section \ref{SEC:SC2}, and giving a summary in Section \ref{SEC:Summary}.

The main contributions in this paper 
%(Sections \ref{SEC:Lifting}, \ref{SEC:Complexity}, \ref{SEC:DegreeTheory} and \ref{SEC:DegreeEvaluation}) 
were presented at ISSAC 2015 \citep{EBD15} and CASC 2016 \citep{ED16a}.  These conference publications addressed the savings in the number of polynomials and their degrees separately.  The present paper unifies the results into a coherent whole providing a single statement of the state of the art.  It also expands on some details, such as the need for primitivity in Section \ref{SEC:Primitive}.  

We include all necessary background theory, allowing the paper to act as a tutorial for CAD with ECs, timely given the increased use of CAD outside of computer algebra systems, as part of satisfiability checkers.  We further expose the results to the wider \textsf{SC$^\mathsf{2}$} by considering implications for CAD in SMT-solvers in Section \ref{SEC:SC2}.  

\section{Background Material}
\label{SEC:Background}

\subsection{CAD computation and terminology}
\label{SUBSEC:CADBackground}

We work under variable ordering $\bm{x} = x_1\prec \ldots \prec x_n$.  The \emph{main variable} of a polynomial or formula (${\rm mvar}$) is the greatest present under the ordering.  

\begin{definition}
A \emph{Cylindrical Algebraic Decomposition (CAD)} is a decomposition of $\mathbb{R}^n$ into connected cells such that:
\begin{itemize}
\item each cell is a \emph{semi-algebraic set} meaning it is defined by a finite sequence of polynomial equations or inequalities; and

\item the cells are arranged \emph{cylindrically} meaning the projections of any two cells in the decomposition on any lower dimensional space with respect to the ordering are either equal of disjoint.
\end{itemize}
\end{definition}
\noindent The latter condition means that each CAD cell is defined by a sequence of conditions: $c_1(x_1), c_2(x_1, x_2), \dots, c_n(x_1, \dots, x_n)$ where each $c_i$ is one of:
\begin{align}
\ell_i(x_1, \dots, x_{i-1}) < &x_i \label{eq:lower}\\
\ell_i(x_1, \dots, x_{i-1}) < &x_i < u_i(x_1, \dots, x_{i-1}) \label{eq:sector}\\
&x_i < u_i(x_1, \dots, x_{i-1}) \label{eq:upper} \\
&x_i = s_i(x_1, \dots, x_{i-1}) \label{eq:section}
\end{align} 
The $\ell_i, u_i, s_i$ are constants when $i=1$ and otherwise most likely an indexed root expression (a particular root of a polynomial).  The former condition tells us that an equivalent semi-algebraic description can be found. 

We describe the computation scheme and terminology that the Collins-descended CAD algorithms share.  Assume a set of input polynomials (possibly derived from formulae).  The first phase of CAD, {\em projection}, applies projection operators repeatedly, each time producing another set of polynomials in one less variable (following the variable ordering).  Together these are the {\em projection polynomials} used in the second phase, {\em lifting}.

First $\mathbb{R}$ is decomposed into cells according to the real roots of polynomials univariate in $x_1$.  Each cell is either a point and therefore its own sample, or an interval inside which we choose a convenient sample (often the dyadic rational with least denominator).  
We next decompose $\mathbb{R}^2$ by repeating a process over each cell in $\mathbb{R}^1$.  In each case we take the bivariate projection polynomials in $(x_1,x_2)$, evaluate them at the  sample point of the cell in $\mathbb{R}^1$ to give univariate polynomials in $x_2$ whose roots we can count and isolate.  

The cells identified in $\mathbb{R}^2$ fall into two categories.  {\em Sections} are defined according to the vanishing of a polynomial as in (\ref{eq:section}), and correspond to the real root of a univariate polynomial.  {\em Sectors} are usually defined as the regions between two sections as in (\ref{eq:sector}), corresponding to the intervals between real roots of a univariate polynomial.  The exceptions are the two infinite sections at either end of the decomposition as in (\ref{eq:upper}) and (\ref{eq:lower}); or if there was no need to decompose at all there may be a single infinite sector with no restrictions on $x_2$.  In each case the sample point is extended from that of the  underlying cell: for sections to include the algebraic number isolated for the real root and for sectors to include any convenient number from the interval (certainly in $\mathbb{Q}$).
Together the sections and sectors form a \emph{stack} over the cell in $\mathbb{R}^1$. 

Taking the union of these stacks gives the CAD of $\mathbb{R}^2$.  The process may then be repeated, each time producing a CAD of larger $\mathbb{R}^i$, until a CAD of $\mathbb{R}^n$ is produced.  The subspaces $\mathbb{R}^i$ decomposed are those implied by the ordering: $(x_1)$- space, $(x_1, x_2)$-space, $(x_1, x_2, x_3)$-space etc.

Cells are represented at a minimum with a sample point and {\em index}.  The latter is a list of integers, with the $k$th describing variable $x_k$ according to the ordered real roots of the projection polynomials.  If the integer is $2i$ the cell is a section corresponds to the $i$th root (counting low to high) and if $2i+1$ it is the adjacent sector\footnote{The dimension of a cell is hence easily identified from the index.}.  % between the $i$th and $(i+1)$th (or the unbounded intervals at either end). 
The projection operator is chosen so polynomials are {\em delineable} in a cell: the portion of their zero set in the cell consists of disjoint sections.  A set of polynomials are {\em delineable} if each is delineable individually, and the sections of different polynomials  are identical or disjoint.  If all projection polynomials are delineable then the input polynomials must be {\em sign-invariant}: have constant sign in each cell of the CAD.  

There have been a great many extensions and optimisations of CAD since its inception, with the survey article by \citet{Collins1998a} highlighting those from the first 20 years.  Since then further highlights have included: symbolic-numeric lifting schemes \citep{Strzebonski2006, IYAY09}; local projection approaches  \citep{Brown2013, Strzebonski2016}; comprehensive Gr\"obner basis approaches \citep{FIS15b} and decompositions via complex space \citep{CMXY09, BCDEMW14}.  

\subsection{Projection operators}
\label{SUBSEC:Proj}

A key improvement to CAD has been in the projection operator to reduce the number of projection polynomials computed \citep{Hong1990,  McCallum1998, McCallum1999b, McCallum2001, Brown2001a, BDEMW13, HDX14, BDEMW16}.  

The minimal complete projection operator (for sign-invariant CAD) proposed is that of \citet{Lazard1994}, however the proof of its correctness was shown to be flawed by \citet{MH16}.  Shortly before this article went to press a corrected proof was published by \citet{MPP19} (requiring changes elsewhere in the CAD lifting phase).  The theory in the present paper uses the family of projection operators by \citet{McCallum1998, McCallum1999b, McCallum2001} but, in due course they may be improved by extending Lazard's family into the EC theory.  We note that the relative savings from the ECs laid out in the present paper would still be maintained if recast into Lazard projection.

Throughout, let $\cont$, $\prim$, $\disc$, $\ldcf$ and $\coeff$ denote the content, primitive part, discriminant, leading coefficient, and set of all coefficients of polynomials respectively (in each case taken with respect to a given mvar).  When applied to a set of polynomials we interpret these as applying the operation to each polynomial in the set.  e.g.
\begin{align*}
\cont(A)  &= \left\{ \cont(f) \, | \, f \in A \right\}, \\
\coeff(A) &= \cup_{f \in A} \coeff(A).
\end{align*}
We let $\res$ denote the resultant of a pair of polynomials, and for a set 
\[
\res(A)=\left\{\res(f_i,f_j) \, | \, f_i \in A, f_j \in A, f_i \neq f_j \right\}.
\]
Recall that a set $B \subset \mathbb{Z}[{\bf x}]$ is an \emph{irreducible basis} if the elements of $B$ are of positive degree in the mvar, irreducible and pairwise relatively prime.  Throughout this section suppose $B$ is an irreducible basis for a set of polynomials, and further that every element of $B$ has mvar $x_n$ and that $F \subseteq B$.

We may now define the projection operators introduced respectively by \citet{McCallum1998, McCallum1999b, McCallum2001}:
\begin{align}
P(B) &:=   \res(B) \cup \disc(B) \cup \coeff(B),
\label{eq:P}
\\
P_{F}(B) &:= P(F) \cup \{ {\rm res}(f,g) \mid f \in F, g \in B \setminus F \},
\label{eq:ECProj} 
\\
P_{F}^{*}(B) &:= P_{F}(B) \cup \disc(B \setminus F) \cup \coeff(B \setminus F).
\label{eq:ECProjStar}
\end{align}
In the general case with $A$ a set of polynomials and $E \subseteq A$ we proceed with projection by: letting $B$ and $F$ be irreducible bases of the primitive parts of $A$ and $E$ respectively; applying the operators as defined above; and then taking the union of the output with $\cont(A)$. 

\begin{remark}
\label{rem:sizes}
We see that (\ref{eq:ECProj}) is contained in (\ref{eq:P}) and will usually be smaller.  It excludes discriminants and coefficients of $B \setminus F$ which are then reinstated by (\ref{eq:ECProjStar}).  It also excludes those resultants which involve two polynomials from $B \setminus F$, an exclusion that is maintained by (\ref{eq:ECProjStar}). Thus we have 
$(\ref{eq:ECProj}) \subseteq (\ref{eq:ECProjStar}) \subseteq (\ref{eq:P})$.
\end{remark}

\begin{remark}
\label{rem:coeffs}
The full set of coefficients of a polynomial is usually unnecessary for CAD \citep{Brown2001a}.  We require knowledge of when the polynomial's degree drops, and thus the vanishing of the leading coefficient is of most importance.  But in the case that it did vanish then the next coefficient becomes leading and thus must also be studied.  To guarantee correctness in all cases the full set is taken but most implementations will safely optimise this.  E.g. if the leading coefficient is constant then it can never vanish and no further coefficients need to considered.
\end{remark}

\begin{remark}
\label{rem:McC01a}
Operator (\ref{eq:ECProjStar}) is different to the $P_{F}^{*}(B)$ of \citet{McCallum2001}, which excluded the coefficients of $B \setminus F$.  The 2001 definition was a mistake, pointed out to us by the anonymous referee of this paper and confirmed in a private communication with McCallum.  However, it is not necessary for us to prove a corrected version of the theorems from that paper because, as was noted by \citet{McCallum2001} (just after Theorem 2.1) if we allow the additional coefficients then we can assume degree invariance and thus use the existing theorems of \citet{McCallum1998} to validate $P_{F}^{*}(B)$.
\end{remark}

\subsection{Sign invariant CAD Projection}
\label{SUBSEC:McC1}

The theorems below use the idea of \emph{order-invariance}, meaning each polynomial has constant order of vanishing within each cell, which of course implies sign-invariance.  We say that a polynomial with mvar $x_k$ is \emph{nullified} over a cell in $\mathbb{R}^{k-1}$ if it vanishes identically throughout.

\begin{theorem}[\citet{McCallum1998}, Thm. 1]
\label{thm:McC1}
Let $S$ be a connected submanifold of $\mathbb{R}^{n-1}$ in which each element of $P(B)$ is order-invariant. 
Then on $S$, each element of $B$ is either nullified or analytic delineable\footnote{A variant on delineability defined by \citet{McCallum1998}.}. Further, the sections of elements of $B$ that are not nullified are either identical or pairwise disjoint, and each element of $B$ is order-invariant on such sections.
\end{theorem}
Suppose we apply $P$ repeatedly to generate projection polynomials.  Repeated use of Theorem \ref{thm:McC1} concludes that a CAD produced by lifting with respect to these projection polynomials is order-invariant so long as no projection polynomial with mvar $x_k$ is nullified over a cell in the CAD of $\mathbb{R}^{k-1}$.  This condition is known as \emph{well-orientedness}.  It is common for problems to be well oriented and the condition can be easily checked for during lifting\footnote{Recall that to lift over a cell we first substitute the cell sample point into the polynomials with main variable one higher: so at this stage we check if any vanish entirely.}.  In the case that a problem is not well-oriented we cannot conclude sign-invariance (at least over the cell in which nullification occurred).  There are some cases where we can rescue the computation (see the report by \citet{Brown2005a}) but in some cases the only option will be to use an alternative complete projection operator, such as that of \citet{Hong1990}. We note that the Lazard operator, recently validated by \citet{MPP19}, removes the well-orientedness condition for sign-invariant CAD\footnote{Instead a modified lifting stage checks for nullification and adapts such polynomials to recover the lost information.}.

\subsection{CAD projection for a formula with a single EC}
\label{SUBSEC:McC2}

A second theorem allows us to understand how $P_E(A)$ is validated.
\begin{theorem}[\citet{McCallum1999b}, Thm. 2.2]
\label{thm:McC2}
Let $f$ and $g$ be integral polynomials with mvar $x_n$ and $r(x_1,\ldots,x_{n-1})$ be their resultant.  Suppose $r \neq 0$.
Let $S$ be a connected subset of $\mathbb{R}^{n-1}$ on which $f$ is delineable and $r$ order-invariant.
Then $g$ is sign-invariant in every section of $f$ over $S$.
\end{theorem}
Suppose $A$ was derived from a formula with EC defined by $E = \{f\}$, and that we apply $P_E(A)$ once and then $P$ repeatedly to generate projection polynomials.  
Assuming the input is well-oriented, we can use Theorem \ref{thm:McC1} to conclude the CAD of $\mathbb{R}^{n-1}$ order invariant for $P_E(A)$.  The CAD of $\mathbb{R}^n$ is then sign-invariant for $E$ using Theorem \ref{thm:McC1} and sign-invariant for $A$ in the sections of $E$ using Theorem \ref{thm:McC2}.  Hence the CAD is truth-invariant for the formula. 

\subsection{CAD projection with multiple ECs}
\label{SUBSEC:McC3}

We now consider the case of multiple ECs.  We could of course apply the previous technology by designating one EC for special use and treating the rest as any other constraint (heuristics can help with the choice \citep{BDEW13}).   But this does not gain any further advantage from the additional ECs, which should be reducing the dimension of our solution space.  It is important to note that we cannot simply add multiple EC defining polynomials into $E$ and use  $P_E(A)$.  That would result in a CAD truth-invariant for the disjunction of the ECs, not the conjunction implied by multiple ECs.

Let us first assume that we have two ECs: one whose mvar is that of the system, $x_{n}$, and another whose mvar is $x_{n-1}$.  Consider applying first the operator $P_E(A)$ where $E$ defines the first EC and then $P_{E'}(A')$ where $A' = P_{E}(A)$ and $E' \subseteq A'$ contains the second EC.  Unfortunately, Theorem \ref{thm:McC2} does not validate this approach.  While it could be applied once for the CAD of $\mathbb{R}^{n-1}$ it cannot then validate the CAD of $\mathbb{R}^{n}$ because the first application of the theorem provided sign-invariance while the second requires the stronger condition of order invariance.  The approach is acceptable if $n=3$ (since in two variables the conditions are equivalent for squarefree bases).  

\begin{example}
\label{ex:Simple}
We consider the formula $\phi = f_1=0 \land f_2=0 \land g \geq 0$ with the following polynomials:
\[
f_1 = x+y^2+z, \quad f_2 = x-y^2+z, \quad g = x^2+y^2+z^2-1.
\]
The polynomials are graphed in Figure \ref{fig:SimpleEx1} where $g$ is the sphere, $f_1$ the upper surface and $f_2$ the lower.  We see that $f_1$ and $f_2$ only meet when $y=0$ and this projection is on the right of Figure \ref{fig:SimpleEx1}.  It shows that the solution requires $|x|\geq\sqrt{2}/2$ and $z=-x$.  How could this be ascertained using CAD?  

\begin{figure}[ht]
\caption{The polynomials from Example \ref{ex:Simple}.}
\label{fig:SimpleEx1}
\begin{center}
\includegraphics[scale=0.45]{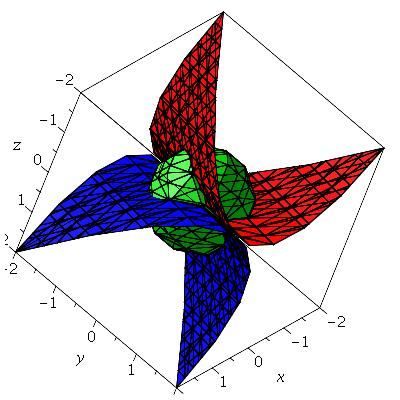}
\includegraphics[scale=0.45]{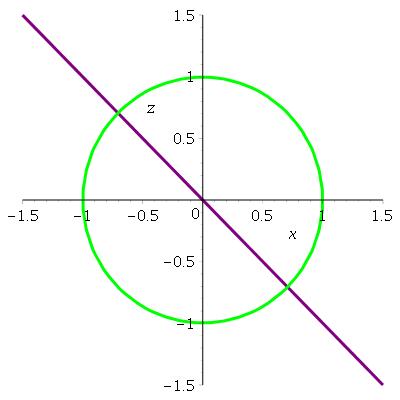}
\end{center}
\end{figure} 

With variable ordering $z \succ y \succ x$ a sign-invariant CAD for $(f_1, f_2, g)$ has 1487 cells using the \textsc{Qepcad-B} program by  \citet{Brown2003b}.
% or  \textsc{Maple}'s \texttt{RegularChains} Library \citep{CMXY09}).  
We could then test a sample point of each cell to identify the ones where $\phi$ is true.  
It is preferable to use the presence of ECs.  Declaring an EC to \textsc{Qepcad} will ensure it uses the algorithm by  \citet{McCallum1999b} based on a single use of $P_E(A)$ followed by $P$.  Either choice results in 289 cells.  In particular, the solution set is described using 8 cells: all have $y=0, z=-x$ but the $x$-coordinate unnecessarily splits cells at $\frac{1}{2}(1\pm\sqrt{6})$.  This is identified due to the projection polynomial $d = \disc_y(\res_z(f_i, g))$.
\end{example}
 
For problems with $n>3$ it is still possible to make use of multiple ECs.  However, we must include the extra information necessary to provide order-invariance of the non-EC polynomials in the sections of ECs.  The following theorem explains that discriminants are required to maintain order invariance, along with degree invariance (and hence coefficients).

%may be used to conclude that $P_{E}^{*}(A)$ is appropriate.
 
%\begin{theorem}[\citet{McCallum2001}, Thm. 2.1]
%\label{thm:McC3}
%Let $f$ and $g$ be integral polynomials with mvar $x_n$, and consider $r = \res(f,g)$,  $d = \disc(g)$.  Suppose $r,d \neq 0$.
%
%Now let $S$ be a connected subset of $\mathbb{R}^{n-1}$ on which $f$ is analytic delineable, $g$ is nowhere nullified and both $r$ and $d$ are order-invariant.
%Then $g$ is order-invariant in each section of $f$ over $S$.
%\end{theorem}

\begin{theorem}[\citet{McCallum1998}, Thm. 2]
\label{thm:McCordinv}
Let f be a polynomial with main variable $x_n$ with positive degree and $d = \disc(f)$ which we suppose to be non-zero.  Let $S$ be a connected submanifold of $\mathbb{R}^{n-1}$ on which $f$ is degree invariant, and does not vanish identically, and in which $d$ is order invariant.  Then $f$ is analytic delineable on $S$ and is order-invariant in each section of $f$ over $S$.
\end{theorem}

Theorem \ref{thm:McCordinv} was used originally as a tool to prove Theorem \ref{thm:McC1}: it gives us the order invariance of polynomials and individual delineability (adding the resultants then extends this to delineability of the set).  We can use it again now to show that the strengthening of (\ref{eq:ECProj}) with the added discriminants and coefficients to form (\ref{eq:ECProjStar}) allows for the order invariance conclusion needed for continued application and validation of the operator.

Suppose we have a formula with two ECs, one with mvar $x_n$ and the other with mvar $x_{n-1}$.   We may now use a reduced operator twice.  We first calculate $A' = P_E(A)$ where $E$ contains the defining polynomial of the first EC, and then $P_{E'}^{*}(A')$ where $E'$ contains the defining polynomial of the other.  Subsequent projections simply use $P$.  When lifting we use: first Theorem \ref{thm:McC1} to verify the CAD of $\mathbb{R}^{n-2}$ as order-invariant for $P_{E'}^{*}(A')$; then Theorem \ref{thm:McC1} to verify the CAD of $\mathbb{R}^{n-1}$ delineable and order-invariant for $E'$, and Theorem \ref{thm:McCordinv} to verify it order-invariant and delineable for the individual polynomials of $A'$ in the sections of $E'$; and finally Theorem \ref{thm:McC1} and \ref{thm:McC2} to verify the CAD of $\mathbb{R}^n$ order-invariant for $E$ and sign-invariant for $A$ in those cells that are both sections of $E$ and $E'$.

%\begin{remark}
%\label{rem:McC01b}
%We cite Theorem \ref{thm:McCordinv} to justify the use of (\ref{eq:ECProjStar}) rather than the corresponding theorem from \citep{McCallum2001}.  As noted in Remark \ref{rem:McC01a} the latter paper contains an omission of the required coefficients.
%\end{remark}

\subsection{EC propagation}
\label{SUBSEC:Propagation}

We can now maximise savings in projection when we have ECs in different main variables.  Further, if we have two ECs with the same mvar we can usually derive another with a lower mvar by taking the resultant of the two defining polynomials.  \citet{McCallum2001} defined this as \emph{EC propagation}.  The process requires the two original ECs to be independent, i.e. the satisfaction of one does not imply the satisfaction of the other.  

Given additional ECs one can perform multiple rounds of propagation to obtain implicit ECs in a sequence of different main variables.  Actually in this case there will be more ECs than we are able to use.  For example, given (independent) ECs $f_i(x,y,z) = 0$ for $i=1,2,3$ in variables $z \prec y \prec x$ then a further three implicit ECs can be found with main variable $y$ and another three with main variable $x$:
\begin{align*}
r_1 &= \res_{z}(f_1, f_2), \quad r_2 = \res_{z}(f_1, f_3), \quad r_3 = \res_{z}(f_2, f_3), \\
R_1 &= \res_{y}(r_1, r_2), \quad R_2 = \res_{y}(r_1, r_3), \quad R_3 = \res_{y}(r_2, r_3).
\end{align*}
Of course, the latter three will not be independent (the vanishing of one should imply the vanishing of another) but even then there may still be questions of efficiency over which to use.  While \citet{BDEW13}h have developed heuristics to help with the choice of which EC to use, there is likely room for improvement\footnote{One possibility is the use of machine learning classifiers to make such choices.  This is a growing topic within mathematical software, with a recent survey given by \citet{England2018}.  It has been applied to CAD by \citet{HEWDPB14, HEDP16}.}.

\begin{example}
\label{ex:Simple1a}
Consider again the example problem from Example \ref{ex:Simple}.  We can propagate the two ECs $f_1=0$ and $f_2=0$ to find implicit EC $r_1=0$ as defined above.  The resultant of the two defining polynomials is $-2y^2$ so we may simplify the EC to $y=0$. 

If we declare both ECs in $z$ to \textsc{Qepcad} then it will perform the propagation for us and use reduced projection twice.  It will actually apply $P_E(A)$ twice (allowed since $n=3$) to produce a CAD with 133 cells.  The solution set is now described using only 4 cells (the minimum possible).  Note that $d$ (see Example \ref{ex:Simple}) was no longer produced as a projection polynomial.
%Our own \textsc{Maple} implementation to apply $P_E(A)$ twice and then lift with respect to all projection polynomials produces 65 cells. 
\end{example}

\section{Reductions in the Lifting Stage}
\label{SEC:Lifting}

The first main contribution of the present paper is to realise that the theorems from the previous section allow for significant savings in the lifting phase (beyond those achieved from reduced projection).  To implement these we must discard two embedded principles of CAD:
\begin{itemize}
\item That the projection polynomials are a fixed set.
\item That the invariance structure of the final CAD can be expressed in terms of sign-invariance of polynomials.
\end{itemize}
The first was also abandoned by \citet{CMXY09, BCDEMW14, JdM12, Brown2015}, while the second by \citet{BM05, MB09}.

\subsection{Minimising the number of polynomials used for lifting}
\label{SUBSEC:IL-Poly}

Consider Theorem \ref{thm:McC2}: it allows us to conclude that $g$ is sign-invariant in the sections of $f$ produced over a CAD of $\mathbb{R}^{n-1}$ order-invariant for $P_{\{f\}}(\{f,g\})$.  Therefore, it is sufficient to perform the final lift with respect to $f$ only (decompose cylinders according to the roots of $f$ but not $g$).  The decomposition imposes sign-invariance for $f$ while Theorem \ref{thm:McC2} guarantees it for $g$ in the cells where it matters (where those signs could change the truth of the formula).  

\begin{example}
\label{ex:Simple2}
We return to Example \ref{ex:Simple}.  Recall that designating either EC and using the algorithm by \citet{McCallum1999b} produced a CAD with 289 cells.  If we follow this approach but lift only with respect to the designated EC at the final step (implemented in the \textsc{Maple} package by \citet{EWBD14}) we obtain a CAD with 141 cells: less than half the output.  
\end{example}

This improved lifting follows from the theorems of \citet{McCallum1999b}, but was only noticed 15 years later during the generalisation of that work to the case of multiple formulae by \citet{BDEMW13, BDEMW16}.  Experiments there demonstrated its importance, particularly for problems with many constraints: see Section 8.3 of \citep{BDEMW16}.  

When we apply a reduced operator at two levels then we can make such reductions at both the corresponding lifts.

\begin{example}
\label{ex:Simple3}
We return to the problem from Example \ref{ex:Simple}.  Set $A = \{f_1, f_2, g\}$ and $E=\{f_1\}$.  The first projection to eliminate $z$ finds
\[
P_E(A) = \{ y, y^4+2xy^2+2x^2+y^2-1 \}.
\]
These are the resultants of $f_1$ with $f_2$ and $g$ (after the first is simplified as discussed in Example \ref{ex:Simple1a}).  The discriminant of $f_1$ was a constant and so could be discarded, as was its leading coefficient (meaning no further coefficients were required as explained in Remark \ref{rem:coeffs}).
Set $A' = P_E(A)$ and $E' = {y}$ (since $y$ defines an EC for the problem as discussed in Example \ref{ex:Simple1a}).  We have $P_{E'}(A') = \{ R \}$ where 
\[
R = \res_y(y, y^4+2xy^2+2x^2+y^2-1) = 2x^2-1.
\]
The other possible entries (the discriminants and coefficients from $E'$) are all constants.  We hence build a 5 cell CAD of the real line with respect to the two real roots of $R$.  We then lift above each cell with respect to $y$ only, in each case splitting the cylinder into three cells about $y=0$, to give a CAD of $\mathbb{R}^2$ with 15 cells.  

Finally, we lift over each of these 15 cells with respect to $f_1$ to give 45 cells of $\mathbb{R}^3$.  This compares to 133 from \textsc{Qepcad}, which used reduced projection but then lifted with \emph{all} projection polynomials.  
No polynomials were nullified, so using Theorems \ref{thm:McC1} and \ref{thm:McC2}, the output is concluded truth-invariant for $\phi$.
\end{example}

\begin{remark}
\label{rem:PointlessLifting}
We note that not only is the additional lifting that \textsc{Qepcad} performed unnecessary for the problem at hand, it also provides no further structure on the output.  For example, if we had lifted with respect to $f_2$ at the final stage in Example \ref{ex:Simple3} then we would be doing so without the knowledge that it is delineable.  Hence splitting the cylinder at the sample point offers no guarantee that the cells produced are sign-invariant away from that sample point.  So the extra work does not allow us to conclude that $f_2$ is sign-invariant (except on sections of $f_1$ where we could have concluded it already).  
\end{remark}

\begin{remark}
\label{rem:WOIssues}
The improvement outlined above not only decreases output size (and computation time) but also the risk of failure\footnote{Although we note that if the recent work of \citet{MPP19} can be extended to ECs then such worries may be unnecessary altogether.}  from non well-oriented input: we only need worry about nullification of polynomials we lift with.
\end{remark}

\subsection{Minimising the cells for stack generation}
\label{SUBSEC:IL-Stack}

We can achieve more savings by abandoning the aim of producing a CAD sign-invariant with respect to any polynomial, instead insisting only on truth-invariance.  We may then lift cells already known to be false trivially to cylinders.  %, only identifying sections of projection polynomials if there is a possibility the formula may be true.  
The idea of avoiding computations over false cells was presented by \citet{Seidl2006}, and one could argue that it is the basis of the Partial CAD for QE problems by \citet{CH91}.  Our contribution here is to explain how such cells can easily be identified in the presence of ECs.  We demonstrate with our example.

\begin{example}
\label{ex:Simple4}
Return to the problem from Examples~\ref{ex:Simple}~$-$~\ref{ex:Simple3} and in particular the CAD of $\mathbb{R}^2$ produced with 15 cells in Example \ref{ex:Simple3}.  On 5 of these 15 cells the polynomial $R$ is zero and on the others it is either positive or negative.  

Now, $\phi$ can only be satisfied above the 5 cells, as elsewhere the two explicit EC defining polynomials cannot share a root and thus cannot vanish together.  We can already conclude the truth value for the 10 cells (false) and thus we do not need to lift over them, except in the trivial sense of extending them to a cylinder in $\mathbb{R}^3$.  Lifting over the 5 cells where $R=0$ with respect to $f_1$ gives 15 cells, which when combined with the 10 cylinders gives a CAD of $\mathbb{R}^3$ with 25 cells that is truth-invariant for $\phi$.
\end{example}

\begin{remark}
\label{rem:NotPointlessLifting}
The improvements in this subsection are affecting the concluded structure of the output\footnote{In comparison with Remark \ref{rem:PointlessLifting}.}.  The final 25 cell truth-invariant CAD in Example \ref{ex:Simple4} is \emph{not} sign-invariant for $f_1$.  The cylinders above the 10 cells in $\mathbb{R}^2$ where $R \neq 0$ may have $f_1$ varying sign, but since $f_1$ can never equal zero at the same time as $f_2$ in these cells it does not affect the truth of $\phi$.  
\end{remark}

Identifying the 5 cells in $\mathbb{R}^2$ where $R=0$ was trivial since they are simply the sections of the second lift: those cells with second entry even in the cell index.  Similarly, all sections in the third lift are those where $f_1$ is zero, however, we cannot conclude that $f_2$ is also zero on these as Theorem \ref{thm:McC2} only guarantees that $f_2$ is sign-invariant on them.  So we must still finish by evaluating the polynomials at the sample points, but only for the sections.

Reducing the number of cells for stack generation clearly decreases output size, and since the cells can be identified using only an integer parity check computation time decreases (less real root isolation is performed).  As described in Remark \ref{rem:WOIssues} for the improvements in Section \ref{SUBSEC:IL-Poly}, this also decreases the risk of non well-oriented input.

\section{Algorithm}
\label{SEC:Alg}

We present Algorithm \ref{alg:ECM} (note that it is split into two parts) to build a truth-invariant CAD for a formula through the use of ECs.  The input is a quantifier free formula in $x_1, \dots, x_n$ (we assume a fixed variable ordering).  The output is a CAD of $\mathbb{R}^n$ which we interpret as a set of cells where each cell comes with a cell index and sample point.  Each of these is a list of $n$ numbers (recall Section \ref{SUBSEC:CADBackground} for how sample points and indices are represented and extended).  This is the minimum information needed, but implementations may choose to store more, such as the formulae for cells. 

The first two steps process the input formula into sets of polynomials: $A_n$ contains all polynomials in the input formula; while the $E_k$ are subsets of $A_n$ which each contain the defining polynomial for a primitive EC of main variable $x_k$ if one is available, and are empty otherwise.  Step \ref{step:extract} is a simple extraction but Step \ref{step:EC} is non-trivial: it must identify suitable ECs for use in projection, and these may not be explicit in the formula (see Definition \ref{def:EC} and Example \ref{ex:EC}).  In practice this will likely require EC propagation (as described in Section \ref{SUBSEC:Propagation}); and EC designation (choosing which of a variety of potential ECs with the same mvar to identify for reduced projection).  We discuss the intricacies of this latter step in more detail later.  

Steps \ref{step:Pstart} $-$ \ref{step:Pend} run the projection phase of the algorithm.  
Each projection starts by identifying contents and primitive parts.  This is not required for $E_i$: since we assume primitive ECs we have $\cont(E_i) = \emptyset$, $\prim(E_i) = E_i$ and the set of factors of $E_i$ contained in $B_i$ for each $i$.

\begin{algorithm}[ht!]\label{alg:ECM}
\caption{CAD using multiple ECs (part 1 of 2)}
%\DontPrintSemicolon
\SetKwInOut{Input}{Input}\SetKwInOut{Output}{Output}
\Input{A QFF $\phi$ in variables $x_1,\ldots,x_n$
}
\Output{Either: $\mathcal{D}$, a truth-invariant CAD of $\mathbb{R}^n$ for $\phi$ formed from a set of cells each defined by an index and a sample point; or 
{\bf FAIL}, if not well-oriented.
}
\BlankLine
Identify from $\phi$ a sequence of sets $E_k$, $k=1, \dots, n$, each either empty or containing a single primitive polynomial with mvar $x_k$, where each polynomial defines an EC for $\phi$\label{step:EC}\;
Extract the set of defining polynomials $A_n$\label{step:extract}\;
%#### Projection ####%
\For{$k = n, \dots, 2$\label{step:Pstart}}{
  Set $B_k$ to the finest squarefree basis for ${\rm prim}(A_k)$\;
  Set $C$ to $\cont(A_k)$\;
  Set $F_k$ to the finest squarefree basis for $E_k$\;
  \eIf{$F_k$ is empty}{
    Set $A_{k-1} := C \cup P(B_k)$\label{step:P}\;
  }{
    \eIf{$k=n$ or $k=2$}{Set $A_{k-1} := C \cup P_{F_i}(B_i)$\;\label{step:PEA}}
    {Set $A_{k-1} := C \cup P_{F_i}^{*}(B_i)$\;\label{step:Pend}}
  }
}
%#### CAD of R1 ####%
If $E_1$ is not empty then set $p$ to be its element; otherwise set $p$ to the product of polynomials in $A_1$\label{step:Bstart}\;  
Build a CAD of the real line, $\mathcal{D}_1$, according to the real roots of~$p$\;
\If{$n=1$}{
\Return $\mathcal{D}_1$\;\label{step:Bend}
}
\end{algorithm}

When there is no declared EC ($E_i$ is empty) the projection operator (\ref{eq:P}) is used (step \ref{step:P}).  Otherwise the operator (\ref{eq:ECProjStar}) is used (step \ref{step:Pend}), unless it is the very first or very last projection (step \ref{step:PEA}) when we use (\ref{eq:ECProj}).  This follows the theory detailed in Section \ref{SEC:Background}.  In each case the output of the projection is combined with the contents to form the next layer of projection polynomials.  

Steps \ref{step:Bstart} $-$ \ref{step:Bend} construct a CAD for the real line (returning it for univariate input), in what is called the \emph{base phase}.  If there is a declared EC in the smallest variable then the real line is decomposed according to its roots; otherwise according to the roots of all the univariate projection polynomials.  

\setcounter{algocf}{0}

\begin{algorithm}[ht!]
\caption{CAD using multiple ECs (part 2 of 2)}
\makeatletter
\advance\c@AlgoLine by 17
\makeatother
%#### Lifting ####%
\For{$k=2, \dots, n$\label{step:Lstart}}{
  Initialise $\mathcal{D}_k$ to be an empty set\;
  \eIf{$F_{k}$ is empty\label{step:L0}}{
    Set $L:=B_k$\;\label{step:L1}
  }{
  Set $L:=F_k$\;\label{step:L2}
  } 
  \eIf{$E_{k-1}$ is empty\label{step:C0}}{
    Set $\mathcal{C}_a := \mathcal{D}_{k-1}$ and $\mathcal{C}_b$ empty\label{step:C1a}\;
  }{
    Set $\mathcal{C}_a$ to be cells in $\mathcal{D}_{k-1}$ whose cell index final entry is even\label{step:C1b}\;
    Set $\mathcal{C}_b := \mathcal{D}_{k-1} \setminus \mathcal{C}_a$\label{step:C2}\;
  }
  \For{each cell $c \in \mathcal{C}_a$}{
    \If{An element of $L$ is nullified over $c$\label{step:WO}}
    {\Return FAIL\;\label{step:fail}
    }
    Generate a stack of cells over $c$ with respect to the polynomials in $L$.  Form new sample points and cell indicies as extensions of those from $c$\label{step:lift}\;
  }
  \For{each cell $c \in \mathcal{C}_b$}{
    Extend to a single cell in $\mathbb{R}^k$ (cylinder over $c$) (the extension to the index is simply 1 and the extension to the sample point can be any number)\label{step:Lend}\;
  }
}
\Return $\mathcal{D}_n$.
\end{algorithm}

Steps \ref{step:Lstart} $-$ \ref{step:Lend} run the lifting phase, incrementally building CADs of $\mathbb{R}^k$ for $k=2, \dots, n$.  For each $k$ there are two considerations:
\newpage
\begin{itemize}
\item First, whether there is a declared EC with mvar $x_k$.  If so we lift only with respect to this (step \ref{step:L2}) and if not we use all projection polynomials with mvar $x_k$ (step \ref{step:L1}).  See Section \ref{SUBSEC:IL-Poly}.
\item Second, whether there is a declared EC with mvar $x_{k-1}$.  If so we restrict stack generation to those cells where the EC was satisfied.  These are simply those with the final entry of the cell index $I_{k-1}$ even (step \ref{step:C1b}).  We lift the other cells trivially to a cylinder in step \ref{step:Lend}.  See Section \ref{SUBSEC:IL-Stack}.
\end{itemize} 
Algorithm \ref{alg:ECM} clearly terminates.  We will verify that it produces a truth-invariant CAD for the formula if the input is well-oriented, as defined below.

\begin{definition}
\label{def:WO}
For $k=2, \dots, n$ define sets:
\begin{itemize}
\item $L_k$ $-$ the \emph{lifting polynomials}: defining polynomial of the declared EC with mvar $x_k$ if it exists; else all projection polynomials with mvar $x_k$.
\item $\mathcal{C}_k$ $-$ the \emph{lifting cells}: those cells in the CAD of $\mathbb{R}^{k-1}$ in which the designated EC with mvar $x_{k-1}$ vanishes if it exists, and all cells in that CAD otherwise. 
\end{itemize}
The input of Algorithm \ref{alg:ECM} is \emph{well-oriented} if for $k=2, \dots, n$ no element of $L_k$ is nullified over an element of $\mathcal{C}_k$.
\end{definition}

\begin{theorem}
\label{thm:Alg}
Algorithm \ref{alg:ECM} satisfies its specification.
\end{theorem}
\begin{proof}
We must show the CAD is truth-invariant for $\phi$, unless the input is not well-oriented when FAIL is returned.

First consider input with $n=1$.  The projection phase would not run, with the algorithm jumping to the CAD construction in step \ref{step:Bstart}, returning the output in step \ref{step:Bend}.  If there was no declared EC then the CAD is sign-invariant for all polynomials defining $\phi$ and thus every cell is truth invariant for $\phi$.  If there was a declared EC then the output is sign-invariant for its defining polynomial.  Cells would either be intervals where the formula must be false; or points, where the EC is satisfied, and the formula either identically true or false depending on the signs of the other polynomials.

Next suppose that the input were not well-oriented (Definition \ref{def:WO}).  For a fixed $k$, the conditional in steps \ref{step:L0} $-$ \ref{step:L2} sets the lifting polynomials $L_k$ to $L$ and the conditional in steps \ref{step:C0} $-$ \ref{step:C2} the lifting cells $\mathcal{C}_{k}$ to $\mathcal{C}_a$.  Thus it is exactly the conditions of Definition \ref{def:WO} which are checked by step \ref{step:WO}, returning FAIL in step \ref{step:fail} when they are not satisfied.  Hence if the lifting phase completes then the input is well-oriented.

From now on we suppose $n>1$ and the input is well-oriented.  
For a fixed $k$ with $2 \leq k \leq n$ define \emph{admissible} cells to be those in the CAD $\mathcal{D}_{k-1}$ of $\mathbb{R}^{k-1}$ produced by Algorithm \ref{alg:ECM} where all declared ECs with mvar smaller than $x_k$ are satisfied, or to be all cells in that CAD if there are no such ECs.  Then let $I(k)$ be the following statement in italics.
% for the CADs $\mathcal{D}_k$ produced by Algorithm \ref{alg:ECM}. 

\emph{Over admissible cells (in $\mathbb{R}^{k-1}$) the CAD $\mathcal{D}_k$ of $\mathbb{R}^k$ produced by Algorithm \ref{alg:ECM} is:
(a) order-invariant for any EC with mvar $x_k$; 
(b) order- (sign- if $k=n$) invariant  for all projection polynomials with mvar $x_k$ on sections of the EC over admissible cells, or over all admissible cells if no EC exists.
}

We shall prove that, for all k with $1 \leq k \leq n$, $I(k)$ is true. We have already proved $I(1)$ (the induction base). Now let $1 < k \leq n$ and assume $I(k-1)$ as the induction hypothesis.  The truth of $I(k)$, which completes the induction, is then a consequence of the following remarks:
\begin{itemizeshort}
\item When $E_k$ is empty we use Theorem \ref{thm:McC1} to assert all projection polynomials with mvar $x_k$ are order-invariant in the stacks over admissible cells giving (a) and (b).  
\item When $E_k$ is not empty and $k=2$ we used the projection operator (\ref{eq:ECProj}).  Theorem \ref{thm:McC2} allows us to conclude (b) and that the EC is sign-invariant in admissible cells.  The stronger property of order-invariance follows since the lifting polynomials form a squarefree basis in two variables.
\item  When $E_k$ is not empty and $k=n$ we used the projection operator (\ref{eq:ECProj}). Theorem \ref{thm:McC2} allows us to conclude (b), but also (a) since in the case $k=n$ the statement requires only sign-invariance.
\item When $E_k$ is not empty and $2<k<n$ we used the projection operator (\ref{eq:ECProjStar}).   Theorem \ref{thm:McCordinv} explains that the additions (\ref{eq:ECProjStar}) makes to (\ref{eq:ECProj}) are sufficient to conclude the statement.
\end{itemizeshort}
In each case the assumptions of the theorems are met by the inductive hypothesis exactly over admissible cells, according to whether $E_{k-1}$ was empty.

From the definition of admissible cells, we know that $\phi$ is false (and thus trivially truth invariant) upon all cells in the CAD of $\mathbb{R}^n$ built over an inadmissible cell of $\mathbb{R}^k$, $k<n$.  Coupled with the truth of (a) for $k=1, \dots, n$, this implies the CAD of $\mathbb{R}^n$ is truth-invariant for the conjunction of ECs (although it may not be truth-invariant for any one individually).  The truth of (b) implies that on those cells where all ECs are satisfied, the other polynomials in $\phi$ are sign-invariant and thus $\phi$ is truth-invariant.
\end{proof}

\section{Worked Example}
\label{SEC:Example}

We consider an example with sufficient variables to show all the features of the algorithm but still small enough to discuss in text.  Assume variable ordering $z \succ y \succ x \succ u \succ v$ and define  
\begin{align*}
&f_1 := x-y+z^2, \qquad f_2 := z^2-u^2+v^2-1, \qquad g := x^2-1, \\
&f_3 := x+y+z^2, \qquad f_4 := z^2+u^2-v^2-1, \qquad h := z.
\end{align*}
We consider the formula
\begin{align*}
\phi = f_1=0 \land f_2=0
\land f_3=0 \land f_4=0 \land g \geq 0 \land h \geq 0.
\end{align*}
The solution can be found manually by decomposing the system into blocks.  The surfaces $f_1$ and $f_3$ are graphed in $(x,y,z)$-space on the left of Figure \ref{fig:WE1}.  They meet only on the plane $y=0$ and this projection is shown on the right.  The surfaces $f_2$ and $f_4$ are graphed in $(z,u,v)$-space on the left of Figure \ref{fig:WE2} and meet only when $z=\pm 1$.  We consider only $z=+1$ due to $h \geq 0$, with this projection plotted on the right.  We thus see that the solution set is 
\[
\{ u=\pm v, x = -1, y = 0, z = 1\}.
\] 

\begin{figure}[p]
\caption{The polynomials $f_1$ and $f_3$ from Section \ref{SEC:Example}.}
\label{fig:WE1}
\begin{center}
\includegraphics[scale=0.48]{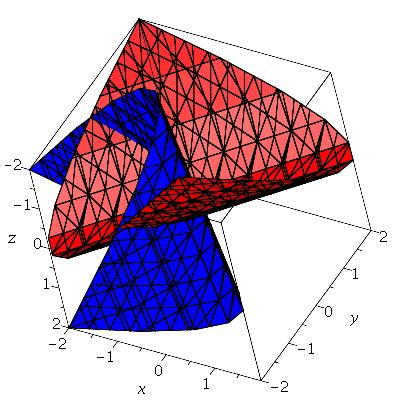}
\includegraphics[scale=0.48]{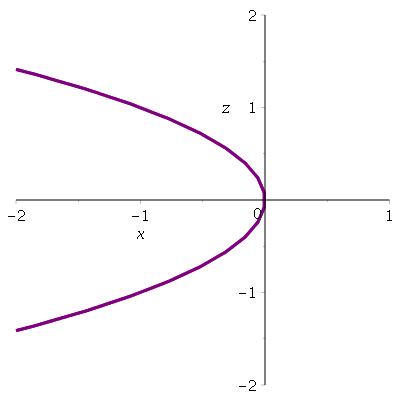}
\end{center}
\end{figure}

\begin{figure}[p]
\caption{The polynomials $f_2$ and $f_4$ from Section \ref{SEC:Example}.}
\label{fig:WE2}
\begin{center}
\includegraphics[scale=0.48]{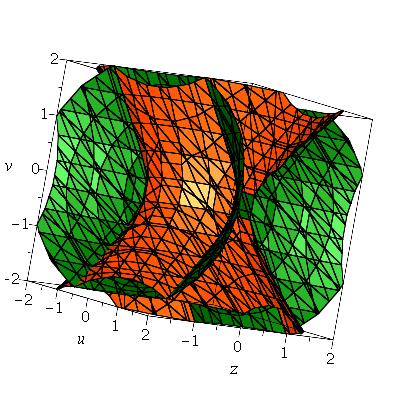}
\includegraphics[scale=0.48]{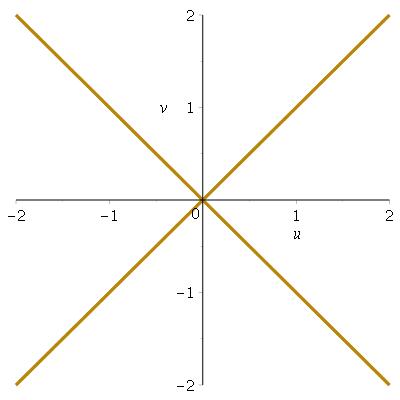}
\end{center}
\end{figure}

To ascertain this by Algorithm \ref{alg:ECM} we must first propagate and designate ECs in Step \ref{step:EC}.  We choose to use $f_1$ first, calculate 
\[
\res_z(f_1, f_2) = (v^2-u^2+y-x-1)^2
\]
and assign $r_1 := v^2-u^2+y-x-1$.  So $r_1$ is the defining polynomial for an EC with mvar $y$.   Similarly consider
\begin{align*}
&\res_y \big( r_1, \res_z(f_1,f_3) \big) = 16(u^2-v^2+x+1)^4, \\
&\res_y \big( r_1, \res_z(f_1,f_4) \big) = 4(u^2-v^2)^2
\end{align*}
and assign $r_2 := u^2-v^2+x+1$, $r_3 := u^2-v^2$.  These are defining polynomials for ECs with mvar $x$ and $u$ respectively.  There is no series of resultants that leads to an EC with mvar $v$ (they all result in constants by that stage).   
We hence identify
$
\{E_j\}_{k=1}^n := \{f_1\}, \{r_1\}, \{r_2\}, \{r_3\}, \{ \, \}
$
in Step \ref{step:EC}.

The algorithm continues by extracting the defining polynomials 
\[
A_5 = \{f_1, f_2, f_3, f_4, g, h\}
\]
and finds $B_5=A_5, F_5=E_5$ (in fact $F_i=E_i$ for all $i=1, \dots, 5$).  

We now start the projection phase.  There is a declared EC for the first projection so we use operator (\ref{eq:ECProj}) to derive
\begin{align*}
A_4 := P_{F_5}(B_5) &= \{(x^2-1)^2, (-u^2+v^2-x+y-1)^2, \\
&\quad (u^2-v^2-x+y-1)^2, 4y^2, x-y\}.
\end{align*}
Hence $C := \{x^2-1\}$ and
\begin{align*}
B_4 &:= \{ y, y-x, -u^2+v^2-x+y-1, u^2-v^2-x+y-1 \}.
\end{align*}
For the next projection we must use operator (\ref{eq:ECProjStar}), giving
\[
A_3 := C \cup P_{F_{4}}^*(B_4) = \{x^2-1, u^2-v^2+x+1, u^2-v^2, u^2-v^2+1\}.
\]
For this example the extra discriminants in (\ref{eq:ECProjStar}) all evaluated to constants and so could be discarded, while for all polynomials the leading coefficient was constant and so could be discarded with no further coefficients considered (see Remark \ref{rem:coeffs}).  Then
\[
B_3 :=  \{x^2-1, u^2-v^2+x+1\}, \quad C := \{u^2-v^2, u^2-v^2+1\},
\]
and the next projection also uses (\ref{eq:ECProjStar}) to produce
\[
A_2 := \{u^2-v^2, u^2-v^2+1, u^4-2u^2v^2+v^4+2u^2-2v^2\}.
\]
For the final projection there is no EC and so we use operator (\ref{eq:P}) to find
$A_1 := \{v^2\}.$ 
The base phase of the algorithm hence produces a 3-cell CAD of the real line isolating $0$.  

For the first lift we have $L = \{u^2-v^2\}$ and $C_a$ containing all 3 cells.  Above the two intervals we split into 5 cells by the curves $u=\pm v$, while above $v=0$ we split into three cells about the origin.
From these 13 cells of $\mathbb{R}^2$ we select the 5 which were sections of $u^2-v^2$ for $C_a$.  These are lifted with respect to $L = \{r_2\}$, and the other 8 are simply extended to cylinders in $\mathbb{R}^3$.  Together this gives a CAD of $\mathbb{R}^3$ with 23 cells.  The next two lifts are similar, producing first a CAD of $\mathbb{R}^4$ with 53 cells and finally a CAD of $\mathbb{R}^5$ with 113 cells.  The entire calculation takes less than a second in \textsc{Maple}.

\subsection{Choice in EC designation}

Algorithm \ref{alg:ECM} could have been initialised with alternative EC designations.  There were the 4 explicit ECs with mvar $z$, and by taking repeated resultants we discover the following implicit ECs, in sets with decreasing mvar:
\begin{align*}
&\{y^2, u^2-v^2+x-y+1, -u^2+v^2+x-y+1, 
\\
&\hspace*{0.25in} u^2-v^2+x+y+1, -u^2+v^2+x+y+1 \},
\\
&\{x+1, -u^2+v^2+x+1, u^2-v^2+x+1\}, \quad
\\ 
&\{u^2-v^2\}.
\end{align*}
There are hence $60$ possible permutations of EC designation, but they lead to only 3 different output sizes:  113, 103 and 93 cells.  
Heuristics for other questions of CAD problem formulation \citep{DSS04, BDEW13, HEWDPB14, WEBD14} could perhaps be adapted to assist here.  None of these are the minimal truth invariant CAD for $\phi$ as all split the CAD of $\mathbb{R}^1$ at $v=0$ (from the discriminant $u^2-v^2$).  

\subsection{Comparison with previous EC theory}

A sign-invariant CAD of $\mathbb{R}^5$ for the 6 input polynomials can be produced by \textsc{Qepcad} with 1,118,205 cells.  Neither the \texttt{RegularChains} Library in \textsc{Maple} \citep{CMXY09} nor our \textsc{Maple} package \citep{EWBD14} could do this in under an hour.  

Our implementation of the algorithm by \citet{McCallum1999b}, which uses operator (\ref{eq:ECProj}) once but also performs the final lift with respect to the EC only, can produce a CAD with either 3023, 10,935 or 48,299 (twice) cells depending on which EC is designated.  The \textsc{Qepcad} implementation of that algorithm gives 11,961, 30,233, 158,475, or 158,451 cells.  Comparing these sets of figures we see the dramatic improvements from just a single reduced lift.

Allowing \textsc{Qepcad} to propagate the four ECs (so a similar projection phase as Algorithm \ref{alg:ECM} but then a normal CAD lifting phase) produces a CAD with 21,079 cells.  By declaring only a subset of the four (which presumably changes the designations of implicit ECs) a CAD with 5,633 cells can be produced, still much more than using Algorithm \ref{alg:ECM}.

The \texttt{RegularChains} Library can also make use of multiple ECs\footnote{but only the latest version from \url{www.regularchains.org}}, as detailed by \citet{BCDEMW14}, a CAD can be produced instantly.  There are choices with analogies to designation \citep{EBCDMW14}), but they all lead to a 137 cell output.  In particular, they all have an induced CAD of the real line which splits at $v = \pm 1$ as well as $v=0$.

We note that our \textsc{Maple} implementation is unrefined and unoptimised.  We do not claim it as a leading CAD implementation.  The purpose of the paper is to illustrate the state of the art in CAD with EC theory, that all CAD implementations should adapt to reproduce.  The worked example shows the clear benefits of the improved lifting techniques, which we next generalise with a complexity analysis.

\section{Complexity Analysis of CAD with EC}
\label{SEC:Complexity}

We build on recent work by \citet{BDEMW16} to measure the dominant term in bounds on the number of CAD cells produced.  Numerous studies have shown this to be closely correlated to the computation time \citep{DSS04, BDEMW13, BCDEMW14}.  We assume CAD input with $m$ polynomials of maximum degree $d$ in any one of $n$ variables.

\begin{definition}
\label{def:md}
Consider a set of polynomials $p_j$.  The \emph{combined degree} of the set is the maximum degree (taken with respect to each variable) of the product of all the polynomials in the set:
$
\textstyle \max_{i} ( \deg_{x_i}( \, \prod_j p_j \, )).
$

The set has the \emph{(m,d)-property} if it may be partitioned into $m$ subsets, each with maximum combined degree $d$.
\end{definition}
For example,  
$
\{y^2-x, y^2+1\}
$
has combined degree $4$ and thus the $(1,4)$-property, but also the $(2,2)$-property.

We will measure complexity by keeping track of the number and degree of projection polynomials.  Of course, by replacing $\{f,g\}$ with $\{fg\}$ we can reduce the number at the cost of increasing the degree, but since it is much easier to find the roots of  $\{f,g\}$ than $\{fg\}$, we do not want to do that.  The $(m,d)$-property, introduced in the thesis of \citet{McCallum1985}, is in some sense the optimal measure of these properties.

\citet{BDEMW16} proved that if $A$ has the $(m,d)$-property then $P(A) \cup \cont(A)$ has the $(M,2d^2)$-property with $M = \left \lfloor \tfrac{1}{2}(m+1)^2 \right\rfloor$. When $m>1$, we can bound $M$ by $m^2$ (but we need $2m^2$ to cover $m=1$).

\subsection{Complexity of sign-invariant CAD}

If $A$ has the $(m,d)$-property then so does its squarefree basis.  Hence applying this result recursively (as in Table \ref{tab:P}) measures the growth in $(m,d)$-property during projection under operator (\ref{eq:P}).  After the first projection there are multiple polynomials and so the tighter bound for $M$ is used. 

\begin{table}[ht]
\caption{Projection under operator (\ref{eq:P}).}  \label{tab:P}
\begin{center}
\begin{tabular}{cccc}
Variables & Number  & Degree                  \\
\hline \hline
$n$       & $m$       & $d$                   \\
$n-1$     & $2m^2$    & $2d^2$                \\
\hline
$n-2$     & $4m^4$    & $8d^4$                \\
%$n-3$     & $16m^8$   & $128d^8$              \\
\vdots    & \vdots    & \vdots                \\
$n-r$     & $2^{2^{r-1}}m^{2^{r}}$   & $2^{2^r-1}d^{2^r}$ \\
\vdots    & \vdots                   & \vdots \\
1         & $2^{2^{n-2}}m^{2^{n-1}}$ & $2^{2^{n-1}-1}d^{2^{n-1}}$
\end{tabular}
\end{center}
\end{table}
The number of real roots in a set with the $(m,d)$-property is at most $md$ (although in practice many will be in $\mathbb{C} \setminus \mathbb{R}$).  
The number of cells in the CAD of $\mathbb{R}^1$ is thus bounded by twice the product of the final two entries, plus 1.  Similarly, if we let $d_i$ and $m_i$ be the entries in the Number and Degree columns of Table \ref{tab:P} from the row corresponding to $i$ variables, then the total number of cells in the CAD of $\mathbb{R}^n$ is bounded by 
\begin{align}
\label{eq:BoundP}
\prod_{i=1}^n \left[2m_id_i + 1\right] = (2md+1) \prod_{r=1}^{n-1} 
\left[ 2 \left( 2^{2^{r-1}}m^{2^r} \right)\left( 2^{2^{r}-1}d^{2^r} \right) + 1 \right].
\end{align}
Omitting the $+1$s will leave us with the dominant term of the bound, which evaluates to give the following result.

\begin{theorem}
\label{thm:Complexity1}
The dominant term in the bound on the number of CAD cells in $\mathbb{R}^n$ produced using (\ref{eq:P}) is
\begin{align}
&\qquad (2d)^{2^{n}-1}m^{2^{n}-1}2^{2^{n-1}-1}.
\label{eq:DominantTerm:P}
\end{align}
\end{theorem}

\subsection{Reduced projection from ECs}

From now on assume $\ell$ equational constraints, $0< \ell \leq \min(m,n)$, all with different mvar.  For simplicity we assume these variables are $x_n, \dots, x_{n-\ell+1}$, i.e. the first $\ell$ projections are the reduced ones.
\begin{lemma}
\label{lem:Complexity}
Suppose $A$ is a set with the $(m,d)$-property and $E \subset A$ has the $(1,d)$-property.  Then $\cont(A) \cup P_{E}^*(A)$ has the $(3m,2d^2)$-property.
\end{lemma}
\begin{proof}
\citet{BDEMW16} proved that applying $P_E(A) \cup \cont(A)$ gives a set of polynomials of size at most $\left \lfloor \tfrac{1}{2}(3m+1) \right\rfloor$ with combined degree $2d^2$.

We now have the additional discriminant and coefficients of (\ref{eq:ECProjStar}) to take care of.  
Each polynomial in $A \setminus E$ will generate an additional discriminant of degree at most $d(d-1)$ and $(d+1)$ additional coefficients of degree at most $d$.  Multiplying all these polynomials together gives a single polynomial of degree at most $2d^2$.  
There are $m-1$ polynomials in $A \setminus E$ and so in total this projection generates 
\[
\left \lfloor \tfrac{1}{2}(3m+1) \right\rfloor + (m-1) < 3m
\]
polynomials of degree $2d^2$.
%The extra $m-1$ discriminants required by operator (\ref{eq:ECProjStar}) will each have degree at most $d(d-1)$, so pairing them we have $\lceil \tfrac{1}{2}(m-1)\rceil$ sets of combined degree at most $2d^2$.  Then
%\[
%\left \lfloor \tfrac{1}{2}(3m+1) \right\rfloor + \lceil \tfrac{1}{2}(m-1)\rceil
%= m + \left \lfloor \tfrac{1}{2}(m+1) \right\rfloor + \lfloor \tfrac{m}{2}\rfloor
%\]
%and since $m \in \mathbb{Z}$ this always equals $2m$.
\end{proof}
We apply this recursively in the top part of Table \ref{tab:P2}, with the bottom derived via the process for $P$, as in Table \ref{tab:P}.   

\begin{table}[ht]
\caption{Projection with (\ref{eq:ECProjStar}) $\ell$ times and then (\ref{eq:P}).}  \label{tab:P2}
\begin{center}
\begin{tabular}{cccc}
Variables    & Number    & Degree                   \\
\hline \hline
$n$          & $m$       & $d$                      \\
$n-1$        & $3m$      & $2d^2$                   \\
%$n-2$        & $9m$      & $8d^4$                   \\
%$n-3$       & $27m$     & $128d^8$                 \\
%\vdots       & \vdots    & \vdots                   \\
%$n-s$        & $3^{s}m$  & $2^{2^{s}-1}d^{2^{s}}$   \\
\vdots       & \vdots    & \vdots                   \\
$n-\ell$     & $3^{\ell}m$    & $2^{2^{\ell}-1}d^{2^{\ell}}$     \\
\hline
$n-(\ell+1)$ & $3^{2\ell}m^2$ & $2^{2^{\ell+1}-1}d^{2^{\ell+1}}$ \\
%$n-(\ell+2)$ & $3^{4\ell}m^4$ & $2^{2^{\ell+2}-1}d^{2^{\ell+2}}$ \\
\vdots       & \vdots                      & \vdots                           \\
$n-(\ell+r)$ & $3^{2^{r}\ell}m^{2^{r}}$    & $2^{2^{\ell+r}-1}d^{2^{\ell+r}}$ \\
\vdots       & \vdots                                   & \vdots              \\
1            & $3^{2^{(n-1-\ell)}\ell}m^{2^{n-1-\ell}}$ & $2^{2^{n-1}-1}d^{2^{n-1}}$ 
\end{tabular}
\end{center}
\end{table}
Define $d_i$ and $m_i$ as the entries in the Number and Degree columns of Table \ref{tab:P2} from the row corresponding to $i$ variables. 
We can bound the number of real roots of projection polynomials in $i$ variables by $m_id_i$.  If we lifted with respect to all projection polynomials the cell count would be bounded by
\begin{align}
&\prod_{i=1}^n \left[2m_id_i + 1\right] 
= \prod_{i=1}^{n-(\ell+1)} \left[2m_id_i + 1\right] 
\times \prod_{i=n-\ell}^{n} \left[2m_id_i + 1\right] 
\label{eq:boundA} \\
&\qquad = \prod_{s=0}^{\ell} \left[ 2 \left( 3^{s}m \right) \left( 2^{2^{s}-1}d^{2^{s}} \right) + 1 \right] 
%\nonumber \\ &\qquad 
\times  \prod_{r=1}^{n-\ell-1} \left[ 
2 \left(  3^{2^{r}\ell}m^{2^{r}} \right) \left(2^{2^{\ell+r}-1}d^{2^{\ell+r}} \right) + 1 
\right]. \nonumber
\end{align}
Omitting the $+1$ from each product allows us to calculate the dominant term of the bound explicitly as
\begin{align}
(2d)^{2^{n}-1} m^{2^{n-\ell} + \ell-1} 3^{\ell 2^{n-\ell} + \ell(\ell-3)/2}.
\label{eq:DominantTerm:P2}
\end{align}

\subsection{Reduced lifting from ECs}

Now we consider the benefit of improved lifting.  
Start by considering the CAD of $\mathbb{R}^{n-(\ell+1)}$.  There can be no reduced lifting until this point and so the cell count bound is given by the second product in (\ref{eq:boundA}), which we will denote by $(\dagger)$.
The lift to $\mathbb{R}^{n-\ell}$ will involve stack generation over all cells, but only with respect to the EC.  
This can have at most $d_{n-\ell}$ real roots and so the CAD at most $(2d_{n-\ell}+1) \times (\dagger)$ cells.

The next lift, to $\mathbb{R}^{n-\ell-1}$, will lift the sections with respect to the EC, and the  sectors only trivially (to produce the same number of cylinders).  Hence the cell count bound is
\[
(2d_{n-(\ell-1)}+1)(d_{n-\ell})(\dagger) + (d_{n-\ell}+1)(\dagger)
\]
with dominant term $2d_{n-(\ell-1)}d_{n-\ell}(\dagger)$.  Subsequent lifts follow the same pattern and so $2d_nd_{n-1} \dots d_{n-(\ell-1)}d_{n-\ell}(\dagger)$ is the dominant term in the bound for $\mathbb{R}^n$.  This evaluates to give the following result.

\begin{theorem}
\label{thm:Complexity2}
Consider the CAD of $\mathbb{R}^n$ produced using Algorithm \ref{alg:ECM} in the presence of ECs in the top $\ell$ variables of the ordering.  The dominant term in the bound on the number of cells is
\begin{align}
&\quad 2 \prod_{s=0}^{\ell} \left[ 2^{2^{s}-1}d^{2^{s}}  \right] 
\prod_{r=1}^{n-\ell-1} \left[ 
2 \left(  2^{2^{r}\ell}m^{2^{r}} 2^{2^{\ell+r}-1}d^{2^{\ell+r}} \right)  
\right] \nonumber \\
&= (2d)^{2^{n}-1} m^{2^{n-\ell} -2} 2^{-\ell} 3^{\ell2^{n-\ell} -2\ell}.
\label{eq:DominantTerm:P3}
\end{align}
\end{theorem}

\subsection{Summary of complexity analysis}

The bound in Theorem \ref{thm:Complexity2} is strictly less than the one in Theorem \ref{thm:Complexity1}. The double exponent of $m$ has decreased by the number of ECs; the result of the improved projection in (\ref{eq:DominantTerm:P2}).  Then improved lifting has reduced the single exponents in the bound further still in (\ref{eq:DominantTerm:P3}).  

However, even with this maximal use of ECs, CAD is still doubly exponential in the number of variables due to the first term in (\ref{eq:DominantTerm:P3}), the one whose degree is the degree term.  This should not be surprising: the theory of ECs is based around reducing the number of polynomials identified in each projection, but not the number of projections which controls the degree growth.  Indeed, we can see directly from Tables \ref{tab:P} and \ref{tab:P2} that at the end of projection we are dealing with univariate polynomials of degree doubly exponential in $n$ regardless of whether we used ECs or not.
Reduced lifting allows us to avoid isolating the real roots of many of these polynomials, but we will always need to consider at least one (the EC defining polynomial).  To control degree growth we must show this to be of a lower degree.

\section{Controlling Degree Growth}
\label{SEC:DegreeTheory}

\subsection{Degree growth through iterated resultant calculations}
\label{SUBSEC-IR}

The doubly exponential degree comes from the use of iterated resultant calculations during projection: the resultant of two degree $d$ polynomials is the determinant of a $2d \times 2d$ matrix whose entries all have degree at most $d$, and thus a polynomial of degree at most $2d^2$.  This increase in degree compounded by $(n-1)$ projections gives the first term of the bound (\ref{eq:DominantTerm:P}).  
Note that the derivation of ECs themselves via EC propagation (Section \ref{SUBSEC:Propagation}) is itself such an iterated resultant calculation.  So even though the EC theory of the previous sections allows us to avoid constructing or lifting with many such polynomials, the ECs themselves encode the degree.

The purpose of the resultant in CAD construction is to ensure that the points in lower dimensional space where polynomials vanish together are identified, and thus that the behaviour over a sample point in a lower dimensional cell is indicative of the behaviour over the cell as a whole.  The iterated resultant (and discriminant) calculations involved in CAD have been studied previously, for example by \citet{McCallum1999a} and \citet{LM09}.  We will follow the work of \citet{BM09} who consider the iterative application of the univariate resultant to multivariate polynomials, demonstrating decompositions into irreducible factors involving the multivariate resultants\footnote{They follow the formalisation of \citet{Jouanolou1991} as laid out in \cite[\S2]{BM09}.}.  They show that the approach will identify polynomials of higher degree than the true multivariate resultant and thus more than required for the purpose of identifying implicit equational constraints.  For example, given 3 polynomials in 3 variables of degree $d$ the true multivariate resultant has degree $\mathcal{O}(d^3)$ rather than $\mathcal{O}(d^4)$.

The key result of \citet{BM09} for our purposes follows.  Note that this considers polynomials of a given \emph{total degree}.  However, the CAD complexity analysis discussed above and later is (following previous work on the topic) with regards to polynomials of \emph{degree at most} $d$ in a given variable.  For clarity we use the Fraktur font when discussing total degree and Roman fonts when the maximum degree.

\def\foo{\citet[Cor. 3.4]{BM09}}
\def\x{\emph{\textbf{x}}}
\def\Res{\mathop{\rm Res}\nolimits}
\begin{corollary}[\foo]
Given three polynomials $f_k(\x,y,z)$ of the form 
\[
f_k(\x,y,z) = \sum_{|\alpha|+i+j\le \mathfrak{d}_k}a_{\alpha,i,j}^{(k)}\x^\alpha y^iz^j \in S[\x][y,z],
\]
where $S$ is any commutative ring, then the iterated univariate resultant 
\[
\res_y\big( \res_z(f_1,f_2),\res_z(f_1,f_3) \big) \in S[\x]
\]
is of total degree at most $\mathfrak{d}_1^2\mathfrak{d}_2\mathfrak{d}_3$ in $\x$, and we may express it in multivariate resultants \citep{Jouanolou1991}, denoted $\Res$, as
\begin{equation}\label{eq:BM}
\begin{array}{c}
\res_y\big( \res_z(f_1,f_2), \res_z(f_1,f_3) \big) = 
(-1)^{\mathfrak{d}_1\mathfrak{d}_2\mathfrak{d}_3} \Res_{y,z}(f_1,f_2,f_3) 
\\ 
\qquad\times
\Res_{y,z,z'}\big( f_1(\x,y,z), f_2(\x,y,z), f_3(\x,y,z'), \delta_{z,z'}(f_1) \big).
\end{array}
\end{equation}
Moreover, if the polynomials $f_1, f_2, f_3$ are sufficiently generic and $n > 1$, then this
iterated resultant has exactly total degree $\mathfrak{d}^2_1\mathfrak{d}_2\mathfrak{d}_3$ in $\x$ and both resultants on the right
hand side of the above equality are distinct and irreducible. 
\end{corollary}

\begin{remark}
Although not stated as part of the result by \citet{BM09}, under these generality assumptions, $\Res_{y,z}(f_1,f_2,f_3)$ has total degree $\mathfrak{d}_1\mathfrak{d}_2\mathfrak{d}_3$ and 
%$Res_{y,z,z'}(f_1(\x,y,z),f_2(\x,y,z),f_3(\x,y,z')\delta_{z,z'}(f_1))$
the second resultant on the right hand side of (\ref{eq:BM}) has total degree $\mathfrak{d}_1(\mathfrak{d}_1-1)\mathfrak{d}_2\mathfrak{d}_3$ (see \cite[Proposition 3.3]{BM09} and \cite[Theorem 2.6]{McCallum1999a}).
\end{remark}

\noindent \citet{BM09} interpret this result in the following quote:\footnote{The quote contains a correction in the description of the second set of roots (removing a dash from $y_1$ in the second distinct point).  The mistake was identified by the anonymous referees of \citep{ED16a}.}. 
\begin{displayquote}
The resultant $r_{12} := \res_z(f_1, f_2)$ defines the projection of the intersection curve between the two surfaces $\{f_1 = 0\}$ and $\{f_2 = 0\}$. 
Similarly, $r_{13} := \res_z(f_1, f_3)$ defines the projection of the intersection curve between the two surfaces $\{f_1 = 0\}$ and $\{f_3 = 0\}$. 
Then the roots of $\res_y(r_{12}, r_{13})$ can be decomposed into two distinct sets: the set of roots $x_0$ such that there exists $y_0$ and $z_0$ such that 
\[
f_1(x_0, y_0, z_0) = f_2(x_0, y_0, z_0) = f_3(x_0, y_0, z_0),
\]
and the set of roots $x_1$ such that there exist two distinct points $(x_1, y_1, z_1)$ and $(x_1, y_1, z'_1)$ such that 
\[
f_1(x_1, y_1, z_1) = f_2(x_1, y_1, z_1) \quad \mbox{and} \quad f_1(x_1, y_1,\allowbreak z '_1) = f_3(x_1, y_1, z'_1).
\] 
The first set gives rise to the term $\Res_{y,z}(f_1, f_2, f_3)$ in the factorization of the iterated resultant
$\res_y(\res_{12}, \res_{13})$, and the second set of roots corresponds to the second factor.
\end{displayquote}
If the $f_i$ are all ECs then only the first set are of interest to us as the truth of the formula of interest needs them all to vanish at once.
However, for a general CAD construction, the second set of roots may also be necessary as they indicate points where the geometry of the sectors changes.

\subsection{How large are these resultants?}

Consider three ECs defined by $f_1, f_2$ and $f_3$ of degree at most $d$ in each variable \emph{separately}; and that we wish to eliminate two variables $z = x_n$ and $y = x_{n-1}$.  We may na\"\i{}vely set each $\mathfrak{d}_i = nd$ to bound the total degree. 

The following approach does better. Let $K=S[x_1,\ldots,x_{n-2},y,z]$ and $L=S[\xi_1,\ldots,\xi_N,y,z]$. 
Only a finite number of monomials in $x_1,\ldots,x_{n-2}$ occur as coefficients of the powers of $y$, $z$ in $f_1$, $f_2$ and $f_3$. Map each such monomial $x^\alpha=\prod_{i=1}^{n-2}x_i^{\alpha_i}$ to $\widetilde{m_j}:=\xi_j^{\max\alpha_i}$ (using a different $\xi_j$ for each monomial\footnote{We could economise: if $x_1x_2^2\mapsto\xi_1^2$, then we could map $x_1^2x_2^4$ to $\xi_1^4$ rather than a new $\xi_2^4$. Since this is for the analysis and not in implementation, we ignore such possibilities.}) and let $\widetilde{f_i}\in L$ be the result of applying this map to the monomials in $f_i$. Operation $\widetilde{\ }$ commutes with taking resultants in $y$ and $z$ (but not $x_i$).

The total degree in the $\xi_j$ of $\widetilde{f_i}$ is the same as the maximum degree in all the $x_1,\ldots,x_{n-2}$ of $f_i$, i.e. bounded by $d$, and hence the total degree of the $\widetilde{f_i}$ in all variables is bounded by $3d$ ($d$ for the $\xi_i$, $d$ for $y$ and $d$ for $z$). If we apply (\ref{eq:BM}) to the $\widetilde{f_i}$, we see that 
$
\res_y\big( \res_z(\widetilde{f_1},\widetilde{f_2}), \res_z(\widetilde{f_1},\widetilde{f_3}) \big)
$
has a factor $\Res_{y,z}(\widetilde{f_1},\widetilde{f_2},\widetilde{f_3})$ of total degree (in the $\xi_j$) $(3d)^3$. Hence, by inverting $\widetilde{\ }$, we may conclude $\Res_{y,z}(f_1,f_2,f_3)$ has maximum degree, in each $x_i$, of  $(3d)^3$.

The results of \citet{Jouanolou1991} and \citet{BM09} apply to any number of eliminations.  In particular, if we have eliminated not $2$ but $\ell-1$ variables we will have a polynomial $\Res_{x_{n-\ell+1}\ldots x_n}(f_{n-\ell},\allowbreak \ldots,f_n)$ of maximum degree $\ell^{\ell}d^{\ell}$ in the remaining variables $x_1,\ldots,x_{n-\ell}$ as the last implicit EC.
Therefore the multivariate resultants we need, $\Res_{x_{n-\ell+1}\ldots x_n}$, only have singly-exponential growth, rather than the doubly-exponential growth of the iterated resultants: can we compute them?

\subsection{Gr\"{o}bner basis instead of iterated resultants}
\label{SUBSEC-GB}

A \emph{Gr\"{o}bner Basis} $G$ is a particular generating set of an ideal $I$ (within the ring of polynomials over an algebraically closed field) defined with respect to a monomial ordering.  One definition is that the ideal generated by the leading terms of $I$ is generated by the leading terms of $G$.  Gr\"obner Bases (GB) allow properties of the ideal to be deduced such as dimension and number of zeros and so are one of the main practical tools for working with polynomial systems.  Their properties and an algorithm to derive a GB for any ideal were introduced in the 1965 PhD thesis of \citet{Buchberger2006} (since republished).  There has been much research to improve and optimise GB calculation, with the $F_5$ algorithm of \citet{Faugere2002} perhaps the most used approach currently. 

Like CAD the calculation of a GB is necessarily doubly exponential in the worst case \citep{MM82} (with lexicographic monomial ordering).  Recent work by \citet{MR13} showed that rather than being doubly exponential with respect to the number of variables present the dependency is in fact on the dimension of the ideal.  Despite this bound GB computation can often be done very quickly usually to the point of instantaneous for any problem tractable by CAD, as demonstrated for example by \citet{WBD12_GB}.

A reasonably common CAD technique is to precondition systems with multiple ECs by replacing the ECs by their GB.  I.e. let $E = \{e_1, e_2, \dots\}$ be a set of polynomials; 
$G = \{g_1, g_2, \dots\}$ a GB for $E$; and $B$ any Boolean combination of constraints, $f_i \, \sigma_i \, 0$, where $\sigma_i \in \{ <, >, \leq, \geq, \neq, =\}$) and $F = \{f_1, f_2, \dots\}$ is another set of polynomials.  Then 
\begin{align*}
\Phi &:= (e_1 = 0 \land e_2 = 0 \land \dots) \land B \quad
\mbox{and} \quad \Psi &:= (g_1 = 0 \land g_2 = 0 \land \dots) \land B
\end{align*}
are equivalent.  A truth-invariant CAD for $\Psi$ is also truth-invariant for $\Phi$.  

If we consider GB preconditioning of CAD in the knowledge of the improved projection schemes for ECs then we see an additional benefit.  It provides implicit ECs which are not in the main variable of the system removing the need for EC propagation.
Since our aim is to produce one EC in each of the last $\ell$ variables, we need to choose an ordering on monomials which is lexicographic with respect to $x_n\succ x_{n-1}\succ\cdots\succ x_{n-\ell+1}$: it does not actually matter (in regards to the theory) how we tie-break after that\footnote{Research suggests that `total degree reverse lexicographic in the rest' is most efficient.}.

Let us suppose that we have $\ell$ ECs $f_1,\ldots,f_{\ell}$ (at least one of them, say $f_1$ must include $x_n$, and similarly we can assume $f_2$ includes $x_{n-1}$ and so on), such that these imply (even over $\mathbb C$) that the last $\ell$ variables are determined (not necessarily uniquely) by the values of $x_1,\ldots,x_{n-\ell}$.  Then the vanishing of polynomials $f_1$, $\Res_{x_n}(f_1,f_2)$, $\Res_{x_n,x_{n-1}}(f_1,f_2,f_3)$ etc. are all implied by the ECs. Hence either they are in the GB, or they are reduced to 0 by the GB, which implies that smaller polynomials are in the GB. Hence our GB will contain polynomials (which are ECs) of degree (in each variable separately) at most
\[
d, \, 4d^2, \, 27d^3, \, \dots, \, ((\ell+1)d)^{\ell+1}.
\]
Note that we are not making, and in the light of the work \citet{MR13} cannot make, any similar claim about the polynomials in fewer variables. Also, it is vital that the ECs be in the last variables for our use of the work of \citet{Jouanolou1991} and \citet{BM09} to work.  So our results do not directly extend from the case we study, first applying $\ell$ reduced CAD projections in the presence of ECs before reverting to standard projection), to the more general case of having any $\ell$ of the projections be reduced.

\subsection{Inclusion in Algorithm \ref{alg:ECM}}

There are two routes to include the above suggestion in Algorithm \ref{alg:ECM}.  

\begin{enumerate}
\item Directly replace the explicit ECs in $\phi$ by those from the GB as suggested above.  This is a pre-processing of the input to Algorithm \ref{alg:ECM}.  The identification of ECs in Step \ref{step:EC} involves only a minimal designation choice when there are multiple explicit ECs in $\phi$ with the same mvar.

\item Encode this process into a sub-algorithm for Step \ref{step:EC}.  The GB polynomials become additional options for designated ECs along with those from EC propagation and choices are made based on minimal degree or some other criteria \citep{WBD12_GB, HEDP16}.  However, in this case, if GB polynomials are designated they must be added to the input set $A_n$ (a reinterpreting of Step \ref{step:extract} so it extracts both from $\phi$ and $\{E_k\}_{k=1}^n$).
\end{enumerate}
The first approach is the one commonly used in implementations (and the one assumed in later discussions).  The benefit of the second is that it caps any increase in the number of polynomials from the use of GBs.

It is unlikely that the GB would produce more polynomials in the main variable than explicit ECs (since we are starting with a generating set all in the main variable and deriving another which would mostly not be) but we have yet to rule it out.  Of course, the number of polynomials in the input can bear little relation to the number generated by projection.  But with the second approach any increase in the initial $m$ is capped to the number of additional designated ECs taken from the GB.  The second option may become preferable in the event of development of a good (cheap) heuristic.

\section{Evaluating the Use of GBs for ECs}
\label{SEC:DegreeEvaluation}

\subsection{Worked Example}
\label{SUBSEC:GBExample}

Let us work with variable ordering $z \succ y \succ x \succ w$; polynomials
\begin{align*}
f_1 &:= xy -z^2 - w^2 + 2z, \qquad
f_2 := x^2+y^2+z^2+w+z, \\
f_3 &:= -w^2-y^2-z^2+x+z \qquad
h := z+w;
\end{align*}
and the following QFF for which we seek a truth-invariant CAD. 
\[
\phi := f_1=0 \land f_2=0 \land f_3 = 0 \land h>0.
\]
In theory, we could analyse this system with a sign-invariant CAD for the four polynomials however none of the CAD implementations in \textsc{Maple} could do this within 30 minutes.  Instead, let us take advantage of the ECs.  There are 3 explicit ECs all with mvar $z$ meaning only one can be designated for the first projection.  We can propagate to find additional implicit ECs:
\begin{align*}
r_1 &= \res_z(f_1, f_2) 
= y^4 + 2xy^3 + (3x^2 -2w^2 + 2w + 6)y^2 + (2x^3 - 2w^2x \\
&\qquad  + 2wx - 3x)y
+ x^4 -2w^2x^2 + 2wx^2 + 6x^2 + w^4 - 2w^3 + 4w^2 + 6w,
\\
r_2 &= \res_z(f_1, f_3) 
= y^4 + 2xy^3 + (x^2-2x+2)y^2 + (x-2x^2)y + w^2 + x^2-2x
\\
r_3 &= \res_z(f_2, f_3)
= 4y^2 + x^4 + 2x^3 -2w^2x^2 + 2wx^2 + 3x^2 -2w^2x + 2wx \\
&\qquad - 2x + w^4 - 2w^3 + 3w^2 + 2w;
\end{align*}
all with main variable $y$.  Continuing the propagation with
\begin{align*}
R_1 := \res_y(r_1, r_2), \qquad 
R_2 := \res_y(r_1, r_3), \qquad 
R_3 := \res_y(r_2, r_3);
\end{align*}
gives the three polynomials in the Appendix, each degree 16 in $x$.  These are different polynomials\footnote{most easily verified by comparing the final lines of each.} but a numerical plot shows them all to have overlapping real part.  All possible resultants to eliminate $x$  evaluate to $0$.

Step \ref{step:EC} could hence produce $3 \times 3 \times 3 = 27$ possible configurations if ECs are identified by propagation.  Our implementation could build CADs for only 6 of these configurations\footnote{The common factor of these 6 was the designation of $r_2$ for second projection.}, when using a time limit of 30 minutes.  Of the 6 completed there was an average of 423 cells calculated in 113 seconds.  The optimal configuration gave 227 cells in 36 seconds using a designation of $f_2, r_3$ and $R_2$.

Now consider instead taking a GB of $\{ f_1, f_2, f_3 \}$.  We use a plex monomial ordering on the same variable ordering as the CAD to achieve a basis:
\begin{align*}
g_1 &= 2z + x^2 + x - w^2 + w, \\
g_2 &= 4y^2 + x^4 + 2x^3 + (-2w^2+2w+3)x^2 + (2w^2+2w-2)x \\
&\qquad + w^4-2w^3+3w^2+2w, \\
g_3 &= 4yx - x^4 - 2x^3 + (2w^2-2w-5)x^2 + (2w^2 - 2w - 4)x 
\\ &\qquad 
- w^4+2w^3-w^2-4w,
\\
g_4 &= (4w^4-8w^3+4w^2+16w)y + x^7 + 4x^6 + (-4w^2+4w+18)x^5 \\
&\qquad + (-12w^2+12w+36)x^4 + (5w^4-10w^3-31w^2+40w+53)x^3 \\
&\qquad + (10w^4-20w^3-34w^2+52w+32)x^2 - (2w^6-6w^5-7w^4+32w^3 \\&\qquad -13w^2-44w-16)x - 2w^6+6w^5-2w^4-14w^3+12w^2+16w,
\\
g_5 &= x^8 + 4x^7 + (-4w^2+4w+18)x^6 + (-12w^2+12w+36)x^5 + (6w^4 \\
&\qquad -12w^3-30w^2+44w+53)x^4 + 4(3w^4-6w^3-8w^2+15w+8)x^3 \\
&\qquad +(-4w^6+12w^5+6w^4-48w^3+26w^2+64w+16)x^2 \\
&\qquad + (-4w^6+12w^5-4w^4-28w^3+24w^2+32w)x \\
&\qquad + w^8-4w^7+6w^6+4w^5-15w^4+8w^3+16w^2.
\end{align*}
This is an alternative generating set for the ideal defined by the explicit ECs and thus all $g_i=0$ are ECs for $\phi$.  Note that the degrees of the GB polynomials (with respect to any one variable) are on average lower (and never greater) than those of the (corresponding) iterated resultants. 

Deriving ECs this way removes the choice for EC with mvar $z$ or $x$ but there are 3 possibilities for the designation with mvar $y$.  Designating $g_2$ yields 83 cells while either $g_3$ or $g_4$ result in 55 cells.  All 3 configurations took less than 20 seconds to compute (with designating $g_4$ the quickest).

\subsection{Effect on the complexity bound}
\label{SUBSEC:GBComplexity}

We now consider how using a GB to produce the designated ECs will improve the complexity analysis of Section \ref{SEC:Complexity}.  The number of polynomials will be the same as found earlier in Table \ref{tab:P2}.  But we must now track separately the degree of the designated EC and the degree of the main projection polynomials as they are derived differently.  For simplicity we will ignore the constant term and focus on the exponents.  

As described above, the designated ECs will have degrees $d, 4d^2, 27d^3, \dots$ as tracked in the middle column of Table \ref{tab:P3}.  For the projection polynomials in the top half of the table the reduced projection operator $P_F(B)$ will take discriminants and coefficients of the EC polynomial; and resultants of them with the other projection polynomials.  Thus the highest degree polynomial produced will have degree that is the sum of the degree of the EC polynomial and the highest degree other polynomial.  This generates the right column of Table \ref{tab:P3}.  We see that the degree exponents here form the so called \emph{Lazy Caterer's sequence}\footnote{The On-Line Encyclopedia of Integer Sequences (2010), Sequence Number A000124, https://oeis.org/A000124} otherwise known as the \emph{Central Polygonal Numbers}.  The remaining projections recorded in the bottom half of the table use the sign-invariant projection operator and so the degree is squared each time.

\begin{table}[ht]
\caption{Maximum degree of projection polynomials produced for CAD when using projection operator (\ref{eq:ECProjStar}) for the first $\ell$ projections and then (\ref{eq:P}) for the remaining.}  \label{tab:P3}
\begin{center}
\begin{tabular}{c@{\hskip 0.1in}|@{\hskip 0.1in}cc}
\multirow{2}{*}{Variables} & \multicolumn{2}{c}{Maximum Degree} \\
                           & EC & Others                \\
\hline \hline
$n$          & $d$       & $d$            \\
$n-1$        & $4d^2$     & $d^2$       \\
$n-2$        & $27d^3$     & $d^4$     \\
$n-3$        & $256d^4$     & $d^7$                 \\
\vdots       & \vdots    & \vdots   \\
$n-\ell$     & $\ell^{\ell}d^{\ell+1}$ & $d^{\ell(\ell+1)(\nicefrac{1}{2}) + 1}$  \\
% & & \\
\hline
% & & \\
$n-(\ell+1)$ & \multicolumn{2}{c}{$d^{\ell(\ell+1) + 2}$}    \\
$n-(\ell+2)$ & \multicolumn{2}{c}{$d^{2\ell(\ell+1) + 2^2}$}    \\
$n-(\ell+3)$ & \multicolumn{2}{c}{$d^{2^2\ell(\ell+1) + 2^3}$}    \\
\vdots       & \multicolumn{2}{c}{\vdots}                  \\
$n-(\ell+r)$ & \multicolumn{2}{c}{$d^{2^{r-1}\ell(\ell+1) + 2^r}$} \\
\vdots       & \multicolumn{2}{c}{\vdots}         \\
1            & \multicolumn{2}{c}{$d^{2^{n-\ell-2}\ell(\ell+1) + 2^{n-\ell-1}}$}
\end{tabular}
\end{center}
\end{table}

Now let us use the the top line of equation (\ref{eq:boundA}) derived earlier as the bound when using improved EC projection and lifting before applying the degrees of the projection polynomials.  We can substitute here with the degrees from Table \ref{tab:P2} as the $d_i$.  The term with base $d$ may be computed by
\[
\textstyle
\prod_{s=0}^{\ell} \big( \ell^{\ell}d^{s+1} \big) \prod_{r=1}^{n-\ell-1} \big( d^{2^{r-1}\ell(\ell+1) + 2^{r}} \big).
\]
The exponent of $d$ evaluates to
\begin{equation}
\label{eq:newDegreeBound}
2^{(n-\ell)}\tfrac{1}{2}(\ell^2 + \ell + 2) - \tfrac{1}{2}(\ell^2 + \ell) - 2.
\end{equation}

\subsubsection{The ignored constants}

Above, we tracked only the degree of the monomial in the bound, and not the constants that multiply it.  As well as for simplicity, this was because we could not find a closed form expression for the product of constants generated.  However, it is simple to check that the constant factors derived by the GB grow exponentially in $\ell$ while those from iterated resultants grow doubly exponentially.  Further, the constant term can be shown to be strictly lower for all but the first few projections.  Finally, note that in Section \ref{SUBSEC-IR} we saw that the multivariate resultant was itself a factor of the iterated resultant.  

\subsubsection{Comparison with base $m$ term}

Let us compare the derived exponent (\ref{eq:newDegreeBound}) with that for the term with base $m$ from (\ref{eq:DominantTerm:P3}): $2^{n-\ell} -2$.  We see that both show the double exponent of the complexity bound reducing by $\ell$, the number of ECs used.  However, the reduction in degree is not quite as clean as the exponential term in the single exponent is multiplied by a quadratic in $\ell$.  This is to be expected as the singly exponential dependency on $\ell$ in the Number column of Table \ref{tab:P} was only in the term with constant base while for Table 2 the term with base $d$ is itself single exponential in $\ell$.

\subsection{Should one always use GBs?}

In Section \ref{SUBSEC:GBExample} we showed the significant savings available if one derived ECs with GBs and in Section \ref{SUBSEC:GBComplexity} we showed this follows through into a theoretical lowering of the worse case complexity bound.  The latter offers the first theoretical justification for what is a widely used CAD optimisation.

However, experimental studies by \citet{BH91, WBD12_GB, HEDP16} have shown that it is not \emph{always} beneficial to pre-process CAD with GB.  
The most recent experiment by \citet{HEDP16} found that 75\% of a data set of 1200 randomly generated CAD problems benefited from GB preconditioning.  So it is certainly worth giving consideration to how ECs are derived.  As noted earlier, the cost of computing the GB itself is usually negligible in comparison to the CAD so it is reasonable to first compute the GB and then decide whether or not to use it.  A simple man-made heuristic was presented by \citet{WBD12_GB} to make the decision while \citet{HEDP16} described the training of a machine learning classifier to decide.

There is no contradiction here with the complexity analysis above:  the analysis is for the worst case and large input and makes no claim to the average complexity or what happens for smaller input.  However, we hypothesise that repeating those studies using the new multiple EC technology would see a reduction in the cases where GB hindered CAD. 

\section{Caveats and the Need for Primitivity}
\label{SEC:Primitive}

There are a few caveats to the results presented above.  First, Algorithm \ref{alg:ECM} can fail for non-well oriented input, but as noted earlier, this restriction may be lifted if the new theory for Lazard's projection operator validated by \citet{MPP19} can be extended to the EC case.  Second, the complexity analysis (both in Section \ref{SEC:Complexity} and \ref{SUBSEC:GBComplexity}) assumes the designated ECs are in strict succession at the start of projection.  For the first analysis it was only made to simplify the working, but for the second analysis it was crucial.  However, Algorithm \ref{alg:ECM} itself does not carry this restriction and savings will still clearly be made in this case.

The only substantial restriction in the paper is that the designated ECs\footnote{whether they be explicit in the formula or calculated via iterated results or GBs.} be defined by primitive polynomials in the main variable of the projection.  The restriction is common in the literature and present in all the underlying theory of \citet{McCallum1999b, McCallum2001}.  

There are analogies to be made with the well-oriented issue (when a projection factor is nullified).  Non-primitive projection factors are not a problem for general CAD because we can factorize prior to projection (order-invariance of factors implies order-invariance of the product).  We cannot do the same for a non-primitive EC though, as the next example shows.  Also, unlike the well-orientedness issue, the primitivity restriction it is not likely to be removed by developing a Lazard family of EC projection operators.  

\subsection{Possibilities to use non-primitive ECs?}
\label{SUBSEC:TTI}

\begin{example}Consider %the formula 
$
\phi := zy = 0 \land \varphi.
$
under ordering $\dots \succ z \succ y \succ \dots$. Polynomial $zy$ is not primitive, so Algorithm \ref{alg:ECM} cannot use the explicit EC.  
\end{example}
We may be tempted to take $E = \{z\}$ as the primitive part, project with operator (\ref{eq:ECProj}) and include the content $y$ in the first projection.  The CAD of $(y, \dots)$-space would be sign-invariant for $y$ and thus the CAD of $(z,y,\dots)$-space truth invariant for the EC (over admissible cells).  But we can no longer say only sections are admissible for the next lift as there may be cells with $z \neq 0$ and $y=0$.  We must instead lift over all cells of $(y, \dots)$-space, saying:
\begin{itemizeshort}
\item Over sections of $y$: $z$ is no longer an EC (as $zy=0$ is forced by $y=0$), so we lift onto all polynomials.
\item Over sectors of $y$: $z$ is an EC, so we only lift with respect to this.
\end{itemizeshort}
In $(z,y, \dots)$-space, some cells are admissible (either $y=0$ or $z=0$) and the rest are not ($zy\ne0$).  %: over the inadmissible cells we just do a trivial lift in higher dimensions.
There are difficulties in forming a general algorithm:%; we list a few:
\begin{enumerate}
\item What would happen if the main variable of the content with respect to $x_i$ were not $x_{i-1}$?
\item What if the content with respect to $x_i$ were itself not primitive as a polynomial in $x_{i-1}$?
\item What if there were another equational constraint in $x_{i-1}$?
\end{enumerate}
The first two can probably (but we have not implemented this yet) be solved by replacing the logic at lines \ref{step:C0}-- in the algorithm by a dynamic determination of which cells were admissible.

Alternatively in that example we might rewrite $\phi$ as
\begin{equation}
\label{eq:TTIphi}
\phi := (z=0 \land \varphi) \lor (y=0 \land \varphi),
\end{equation}
so each clause has its own EC. 
The theory of truth-table invariant CADs (TTICADs) developed by \citet{BDEMW13, BDEMW16} is designed to deal with such input. More generally, given a formula of the form 
\begin{equation}\label{eq:TTI1}
(f_1=0\land g_1>0)\lor(f_2=0\land g_2>0),
\end{equation}
TTICAD allows for an improvement on the standard EC theory.  Since $f_1f_2=0$ is an implicit EC of (\ref{eq:TTI1}) standard EC theory allows us to avoid studying the $g_i$ away from where any $f_i$ is zero.  By utilising a TTICAD from \citep{BDEMW13} we can also avoid studying the $g_i$ away from where the \emph{corresponding} $f_i$ is zero.  This approach was extended in \citep{BDEMW16} to also consider formulae such as 
\begin{equation}\label{eq:TTI2}
(f_1=0\land g_1>0)\lor(f_2>0\land g_2>0),
\end{equation}
where there is no single implicit EC.  
Although (\ref{eq:TTIphi}) looks closer to (\ref{eq:TTI1}) it is actually more like (\ref{eq:TTI2}) since $y$ is not an EC in the main variable.  

Although there is the possibility of applying TTICAD for this problem it would first require its own extension to use beyond the first projection (analogous to the present work for standard ECs).  %Of course, this would be a valuable extension in its own right.

\subsection{Classical non primitivity}
\label{SUBSEC:DHExamples}

We can see the importance of the primitivity restriction in the classic complexity results of \citet{BD07}, \citet{DH88}.  Both rest on the following construction.
Let $P_k(x_k,y_k)$ be the statement $x_k=f(y_k)$ and then define recursively
\begin{align}
&P_{k-1}(x_{k-1},y_{k-1}) := \label{eq:H} \\
&\begin{array}{c}
\cr\underbrace{\exists z_k\forall x_k\forall y_k}_{Q_k}\underbrace{\left((y_{k-1}=y_k\land x_{k}=z_k)\lor(y_{k}=z_k\land x_{k-1}=x_k)\right)}_{L_k}\Rightarrow P_k(x_k,y_k).
\end{array} \nonumber
\end{align}
%\begin{equation}
%\label{eq:H}
%\begin{array}{c}
%P_{k-1}(x_{k-1},y_{k-1}):=\cr\underbrace{\exists z_k\forall x_k\forall y_k}_{Q_k}\underbrace{\left((y_{k-1}=y_k\land x_{k-1}=z_k)\lor(y_{k-1}=z_k\land x_{k-1}=x_k)\right)}_{L_k}\Rightarrow P_k(x_k,y_k).
%\end{array}
%\end{equation}
This is $\exists z_k \left(z_k=f(y_{k-1})\land x_{k-1}=f(z_k)\right)$, i.e. $x_{k-1}=f(f(y_{k-1}))$.  
%
%\newpage
%
\noindent Repeated nesting of this procedure builds the doubly-exponential growth.  So  
\begin{equation}\label{eq:H2}
P_{k-2}(x_{k-2}, y_{k-2}) = Q_{k-1}L_{k-1}\Rightarrow \left(Q_kL_k\Rightarrow P_k(x_k,y_k)\right),
\end{equation}
gives $x_{k-2}=f(f(f(f(y_{k-2}))))$ etc.
Rewriting (\ref{eq:H2}) in prenex form gives 
\begin{equation}\label{eq:H3}
P_{k-2}(x_{k-2}, y_{k-2}) = Q_{k-1}Q_k \neg L_{k-1}\lor\neg L_k\lor P_k(x_k,y_k).
\end{equation}
The negation of  (\ref{eq:H3}) is therefore
\begin{equation}\label{eq:H4}
\neg P_{k-2}(x_{k-2}, y_{k-2}) = \overline Q_{k-1}\overline Q_k  L_{k-1}\land L_k\land\neg P_k(x_k,y_k),
\end{equation}
where the $\overline{\strut\quad}$ operator interchanges $\forall$ and $\exists$.
Now, $L_k$ can be rewritten as
\begin{align}
\label{eq:H5}
L_k &= (y_{k-1} = y_k\lor y_{k}=z_k) \land (y_{k-1}=y_k\lor x_{k-1}=x_k) \nonumber \\
&\qquad \land (x_{k}=z_k\lor y_{k}=z_k)\land(x_{k}=z_k\lor x_{k-1}=x_k)
\end{align}
and further
\begin{align}
\label{eq:H6}
L_k &= (y_{k-1}-y_k)( y_{k}-z_k)=0\land(y_{k-1}-y_k)( x_{k-1}-x_k)=0 \nonumber \\
&\qquad \land (x_{k}-z_k)( y_{k}-z_k)=0\land(x_{k}-z_k)( x_{k-1}-x_k)=0,
\end{align}
which shows $L_k$ to be a conjunction of (non primitive) ECs. This is true for any $L_i$, hence  the propositional part of (\ref{eq:H4}) is a conjunction of eight ECs, mostly non primitive, and $\neg P_k(x_k,y_k)$.  Hence by induction we have that the whole family of examples $\neg P_i$ may be written as complete conjunction of (mostly non primitive) ECs.  
Furthermore there are equalities whose main variables are the first variables to be projected if we try to produce a quantifier-free form of (\ref{eq:H4}). But that quantifier-free form describes the complement of the semi-algebraic varieties in \citep{BD07} or \citep{DH88}  (depending which $P_k$ we take) and these have doubly-exponential complexity in $n$.

So we observe that the classical results proving the doubly exponential complexity of CAD are not tackled by our EC technology.

\section{Lessons for \textsf{SC$^\mathsf{2}$}}
\label{SEC:SC2}

The \textsf{SC$^\mathsf{2}$} community already appreciates that the logical structure of CAD input is important and should be exploited where ever possible.  The main additional lesson from the present paper is that this exploitation can take place not only at the Boolean skeleton level but also in the computer algebra.

\subsection{Reasons for optimism}

There are high barriers to implementing CAD without the support of a computer algebra system, however, SMT solvers such as \textsc{SMT-RAT} by \citet{LSCAB13, KA19} and \textsc{Z3}  by \citet{JdM12} show it is possible.  Indeed, the developers of \textsc{SMT-RAT} are now beginning to expand their CAD module to include a variety of projection operators \citep{VKA17} and even EC functionality as described by \citet{HKA18}.  

One particular barrier is the need multivariate factorization algorithms, which those developing in a computer algebra system can take for granted but represent a significant implementation cost.  On this point we highlight to the SMT community the availability of \textsc{CoCoALib} which is a free C++ library that can perform computer algebra computations without the requirement for the accompanying Computer Algebra System (in this case \textsc{CoCoA}) \citep{AB14}.  Further, \textsc{CoCoA} is now actively developing features for use in SMT as described by \citet{AB17, ABP18}.

\subsection{Incrementality}

A key requirement for the effective use of CAD by SMT-solvers is that the CAD technology be incremental: that polynomials can be added and removed to the input with the data structures of the CAD edited rather than recalculated.  Such incremental CAD algorithms are now under development as part of the \textsf{SC$^\mathsf{2}$} by \citet{KA19, CE18}.

An additional advantage from incremental CAD would be with regards to the issue of well-orientedness.  I.e. if a particular operator is found to not be well-oriented at the end of a CAD calculation the next step would be to revert to a less efficient operator which is a superset of the original.  Refining an existing decomposition should be cheaper than recomputing from scratch.  Although on this point, the development of the Lazard projection theory may remove the well-orientedness condition all together.

However, the use of CAD with ECs incrementally requires additional development work.  First, it introduces additional decisions to be taken such as EC designation and whether to pre-processing with GB (not to mention whether that can also be done incrementally).  Second, this growing number of decisions needs to be taken in tandem, prompting exponential growth in the number of possibilities that overwhelms existing heuristics.  Machine learning techniques may be one way forward, as outlined by \citet{England2018}.

Finally, existing heuristics that guide the Boolean search may not be suitable since the use of ECs could prompt what appears as strange behaviour in the SMT context.  For example, removing a constraint that was equational could actually grow the output CAD since it necessitates the use of a larger projection operator.  Correspondingly, adding an equational constraint could allow a smaller operator and shrink the output.  SMT solver search heuristics will need to be adapted to handle these possibilities.

\newpage

\section{Summary}
\label{SEC:Summary}

We have presented much of the state of the art in the theory of CAD with Equational Constraints.  This included how ECs may be leveraged for savings in the lifting phase as well as projection.  We demonstrated the benefits of the theory with worked examples and complexity analysis.  The latter shows that the worst CAD bound has double exponent that reduced from $n$ by the number of ECs.  Crucially, this is the global double exponent covering both the number and degree of polynomials, if we allow for Groebner Basis pre-processing.  

The main avenues for future work are an exploration of dealing with non-primitive ECs; the extension of the Lazard projection operator to a family of operators for ECs; the development of heuristics for choosing which ECs to designate; and the development of incremental EC technology.  We note that the current results and any future progress have benefits not only for Symbolic Computation but the wider \textsf{SC$^\mathsf{2}$} community.

\subsection*{Acknowledgements}

This work was originally supported by EPSRC grant EP/J003247/1 and later by EU H2020-FETOPEN-2016-2017-CSA project $\mathcal{SC}^2$ (712689). 

We are grateful to all the anonymous referees of this paper, and also those of our conference papers at ISSAC 2015 \citep{EBD15} and CASC 2016 \citep{ED16a}, for many helpful comments.  We are particularly grateful to the referee who pointed out the mistake in the literature on operator (\ref{eq:ECProjStar}) and Dr McCallum for discussing this with us.  

We also thank Prof. Buchberger for reminding JHD that Gr\"obner bases were applicable to the problem of degree growth.




\appendix

\section{The Iterated Resultants From Section \ref{SUBSEC:GBExample}}

%The iterated resultants discussed in Section \ref{SEC-Example} are as follows.

\begin{footnotesize}

\begin{align*}
R_1 &:= \res(r_1, r_2, y) 
=  
{x}^{16}
+ 8\,{x}^{15}
+ ( -8\,{w}^{2}+8\,w+64 ) {x}^{14}
+ ( -56\,{w}^{2}+56\,w \\ &\quad +288 ) {x}^{13}
+ ( 28\,{w}^{4}-56\,{w}^{3}-332\,{w}^{2}+400\,w+1138 ) {x}^{12}
+ ( 168\,{w}^{4} \\ &\quad -336\,{w}^{3}-1144\,{w}^{2}+1552\,w+2912 ) {x}^{11}
+ ( -56\,{w}^{6}+168\,{w}^{5}+648\,{w}^{4} \\ &\quad -1816\,{w}^{3}-2664\,{w} ^{2}+5328\,w+6336 ) {x}^{10}
+ ( -280\,{w}^{6}+840\,{w}^{5} \\ &\quad +1400\,{w}^{4}-5400\,{w}^{3}-2616\,{w}^{2}+11368\,w+7808 ) {x}^{9}
+ ( 70\,{w}^{8} \\ &\quad -280\,{w}^{7}-500\,{w}^{6}+3080\,{w}^{5}-270\,{w}^{4}-11576\,{w}^{3}+4860\,{w}^{2} \\ &\quad +20816\,w+7381 ) {x}^{8}
+ ( 280\,{w}^{8}-1120\,{w}^{7}+80\,{w}^{6}+6080\,{w}^{5}-8480\,{w}^{4} \\ &\quad -11792\,{w}^{3}+22840\,{w}^{2}+20192\,w+920 ) {x}^{7}
+ ( -56\,{w}^{10}+280\,{w}^{9} \\ &\quad -80\,{w}^{8}-2160\,{w}^{7}+4960\,{w}^{6}+3200\,{w}^{5}-22608\,{w}^{4}+2584\,{w}^{3} \\ &\quad +40840\,{w}^{2}+16040\,w+2024 ) {x}^{6}
+ ( -168\,{w}^{10}+840\,{w}^{9}-1520\,{w}^{8} \\ &\quad -1360\,{w}^{7}+12016\,{w}^{6}-11296\,{w}^{5}-23368\,{w}^{4}+30136\,{w}^{3}+22032\,{w}^{2} \\ &\quad +624\,w+736 ) {x}^{5}
+ ( 28\,{w}^{12}-168\,{w}^{11}+396\,{w}^{10}+160\,{w}^{9}-3690\,{w}^{8} 
\\
&\quad +6576\,{w}^{7}+4520\,{w}^{6}-24712\,{w}^{5}+13154\,{w}^{4}+37456\,{w}^{3}+1464\,{w}^{2} 
\\
&\quad -1568\,w+5968 ) {x}^{4}
+ ( 56\,{w}^{12}-336\,{w}^{11}+1192\,{w}^{10}-1680\,{w}^{9} 
\\
&\quad -2688\,{w}^{8}+12496\,{w}^{7}-13464\,{w}^{6}-16912\,{w}^{5}+37240\,{w}^{4}+13472\,{w}^{3} \\ &\quad -16384\,{w}^{2}+1984\,w+3072 ) {x}^{3}
+ ( -8\,{w}^{14}+56\,{w}^{13}-248\,{w}^{12}+520\,{w}^{11} \\ 
&\quad +72\,{w}^{10}-3088\,{w}^{9}+7664\,{w}^{8}-2040\,{w}^{7}-16176\,{w}^{6}+20424\,{w}^{5} \\ &\quad +20056\,{w}^{4}-15360\,{w}^{3}-8544\,{w}^{2}+4608\,w+2304 ) {x}^{2}
+ ( -8\,{w}^{14} \\ &\quad +56\,{w}^{13}-296\,{w}^{12}+808\,{w}^{11}-1144\,{w}^{10}-776\,{w}^{9}+6184\,{w}^{8}-7048\,{w}^{7} \\ 
&\quad -6944\,{w}^{6}+19696\,{w}^{5}+3872\,{w}^{4}-16832\,{w}^{3}-1152\,{w}^{2}+4608\,w ) x + {w}^{16} \\ 
&\quad 
-8\,{w}^{15}+52\,{w}^{14}-184\,{w}^{13}+454\,{w}^{12}-440\,{w}^{11}-772\,{w}^{10} +3352\,{w}^{9}
\\ &\quad -2447\,{w}^{8}-4288\,{w}^{7}+8200\,{w}^{6}+2080\,{w}^{5}-7664\,{w}^{4} -384\,{w}^{3} + 2304\,{w}^{2}
\\
R_2 &:= \res(r_1, r_3, y) =  
{x}^{16} + 8\,{x}^{15}
+ ( -8\,{w}^{2}+8\,w+28 ) {x}^{14}
+ ( -56\,{w}^{2}+56\,w \\ 
&\quad +48 ) {x}^{13}
+ ( 28\,{w}^{4}-56\,{w}^{3}-116\,{w}^{2}+160\,w-2 ) {x}^{12}
+ ( 168\,{w}^{4} \\ &\quad -336\,{w}^{3}+80\,{w}^{2}+184\,w-256 ) {x}^{11}
+ ( -56\,{w}^{6}+168\,{w}^{5}+108\,{w}^{4} \\ 
&\quad -592\,{w}^{3}+852\,{w}^{2}-240\,w-12 ) {x}^{10}
+ ( -280\,{w}^{6}+840\,{w}^{5}-1120\,{w}^{4} \\ 
&\quad +360\,{w}^{3}+1872\,{w}^{2}-1448\,w+2000 ) {x}^{9}
+ ( 70\,{w}^{8}-280\,{w}^{7}+220\,{w}^{6} \\ 
&\quad +560\,{w}^{5}-2742\,{w}^{4}+3232\,{w}^{3}-1428\,{w}^{2}+224\,w+4537 ) {x}^{8}
+ ( 280\,{w}^{8} \\ 
&\quad -1120\,{w}^{7}+2720\,{w}^{6}-3280\,{w}^{5}-1280\,{w}^{4}+6016\,{w}^{3}-11696\,{w}^{2}+7496\,w \\
 &\quad +2552 ) {x}^{7}
+ ( -56\,{w}^{10}+280\,{w}^{9}-620\,{w}^{8}+480\,{w}^{7}+2488\,{w}^{6}-6880\,{w}^{5} 
\\
&\quad +9384\,{w}^{4}-5744\,{w}^{3}-9404\,{w}^{2}+12008\,w-4120 ) {x}^{6}
+ ( -168\,{w}^{10}
\end{align*}
\end{footnotesize}
\begin{footnotesize}
\begin{align*} 
&\quad +840\,{w}^{9} -2960\,{w}^{8}+5840\,{w}^{7}-4832\,{w}^{6}-3088\,{w}^{5}+21104\,{w}^{4} 
\\ &\quad -27128\,{w}^{3} +12552\,{w}^{2}+3888\,w-5888 ) {x}^{5}
+ ( 28\,{w}^{12}-168\,{w}^{11}+612\,{w}^{10} \\ &\quad -1280\,{w}^{9} +498\,{w}^{8}+3648\,{w}^{7}-12424\,{w}^{6}+17360\,{w}^{5}-4546\,{w}^{4} \\ 
&\quad -13928\,{w}^{3} + 19032\,{w}^{2}-9344\,w-176 ) {x}^{4}
+ ( 56w^{12} - 336w^{11} + 1552w^{10} \\ &\quad - 4200\,{w}^{9} + 7296\,{w}^{8} - 6080\,{w}^{7} - 7440\,w^{6}+25880\,{w}^{5}-31352\,{w}^{4} \\ &\quad +13472\,{w}^{3} +1856\,{w}^{2}-10304\,w+1536 )x^3
+ ( -8\,{w}^{14}+56\,{w}^{13}-284\,{w}^{12} \\ &\quad +880\,{w}^{11} -1740\,{w}^{10}+1616\,{w}^{9}+2468\,{w}^{8}-10704\,{w}^{7}+15828\,{w}^{6} \\ &\quad -8040\,{w}^{5} -1064\,{w}^{4}+9792\,{w}^{3}-3168\,{w}^{2}+2304 ) {x}^{2}
+ ( -8\,{w}^{14}+56\,{w}^{13} \\ &\quad -320\,{w}^{12} +1096\,{w}^{11}-2800\,{w}^{10}+4600\,{w}^{9}-3968\,{w}^{8}-2152\,{w}^{7} \\ &\quad +9592\,{w}^{6} -10832\,{w}^{5} +5312\,{w}^{4}+4672\,{w}^{3}-5760\,{w}^{2}+4608\,w ) x
+ {w}^{16} \\ &\quad -8\,{w}^{15} +52\,{w}^{14}-208\,{w}^{13} +646\,{w}^{12}-1376\,{w}^{11}+2012\,{w}^{10}-1136\,{w}^{9} \\ &\quad -1295\,{w}^{8}+4328\,{w}^{7}-3992\,{w}^{6} +2368\,{w}^{5}+2320\,{w}^{4}-1920\,{w}^{3}+2304\,{w}^{2}
\\
R_3 &:= \res(r_3, r_3, y) =  
{x}^{16} + 8\,{x}^{15}
+ ( -8\,{w}^{2}+8\,w+44 ) {x}^{14}
+ ( -56\,{w}^{2}+56\,w  
\\
&\quad +160 ) {x}^{13}
+ ( 28\,{w}^{4}-56\,{w}^{3}-228\,{w}^{2}+272\,w+430 ) {x}^{12}
+ ( 168\,{w}^{4}
\\
&\quad -336\,{w}^{3}-592\,{w}^{2}+856\,w+816 ) {x}^{11}
+ ( -56\,{w}^{6}+168\,{w}^{5}+444\,{w}^{4}
\\
&\quad -1264\,{w}^{3}-812\,{w}^{2}+1952\,w+1092 ) {x}^{10}
+ ( -280\,{w}^{6}+840\,{w}^{5}+560\,{w}^{4}
\\
&\quad -3000\,{w}^{3}+32\,{w}^{2}+3032\,w+736 ) {x}^{9}
+ ( 70\,{w}^{8}-280\,{w}^{7}-340\,{w}^{6}
\\
&\quad +2240\,{w}^{5}-902\,{w}^{4}-4208\,{w}^{3}+2716\,{w}^{2}+3120\,w-183 ) {x}^{8}
+ ( 280\,{w}^{8}
\\
&\quad -1120\,{w}^{7}+480\,{w}^{6}+3440\,{w}^{5}-4640\,{w}^{4}-2304\,{w}^{3}+5840\,{w}^{2}+1128\,w\\ &\quad -1144 ) {x}^{7}
+ ( -56\,{w}^{10}+280\,{w}^{9}-60\,{w}^{8}-1760\,{w}^{7}+3128\,{w}^{6}+960\,{w}^{5}\\ &\quad -7352\,{w}^{4}+3216\,{w}^{3}+5860\,{w}^{2}-1320\,w-824 ) {x}^{6}
+ ( -168\,{w}^{10}+840\,{w}^{9}\\ &\quad -1280\,{w}^{8}-880\,{w}^{7}+5568\,{w}^{6}-5008\,{w}^{5}-4464\,{w}^{4}+7848\,{w}^{3}+984\,{w}^{2}\\ &\quad -2576\,w-64 ) {x}^{5}
+ ( 28\,{w}^{12}-168\,{w}^{11}+276\,{w}^{10}+400\,{w}^{9}-2302\,{w}^{8}\\ &\quad +2848\,{w}^{7}+1880\,{w}^{6}-7440\,{w}^{5}+3582\,{w}^{4}+5704\,{w}^{3}-3208\,{w}^{2}-1216\,w\\ &\quad +720 ) {x}^{4}
+ ( 56\,{w}^{12}-336\,{w}^{11}+880\,{w}^{10}-840\,{w}^{9}-1424\,{w}^{8}+4800\,{w}^{7}\\ &\quad -3856\,{w}^{6}-3464\,{w}^{5}+6968\,{w}^{4}+32\,{w}^{3}-3392\,{w}^{2}+448\,w+512 ) {x}^{3}
\\ &\quad + ( -8\,{w}^{14}+56\,{w}^{13}-172\,{w}^{12}+208\,{w}^{11}+308\,{w}^{10}-1504\,{w}^{9}+1972\,{w}^{8}\\ &\quad +432\,{w}^{7}-3788\,{w}^{6}+2920\,{w}^{5}+2552\,{w}^{4}-3136\,{w}^{3}-864\,{w}^{2}+1024\,w\\ &\quad +256 ) {x}^{2}
+ ( -8\,{w}^{14}+56\,{w}^{13}-208\,{w}^{12}+424\,{w}^{11}-352\,{w}^{10}-520\,{w}^{9}\\ &\quad +1744\,{w}^{8}-1416\,{w}^{7}-1176\,{w}^{6}+2928\,{w}^{5}-384\,{w}^{4}-1984\,{w}^{3}+384\,{w}^{2}\\ &\quad +512\,w ) x
+ {w}^{16}-8\,{w}^{15}+36\,{w}^{14}-96\,{w}^{13}+150\,{w}^{12}-48\,{w}^{11}-308\,{w}^{10}\\ &\quad +672\,{w}^{9}-351\,{w}^{8}-648\,{w}^{7}+1096\,{w}^{6}-880\,{w}^{4}+128\,{w}^{3}+256\,{w}^{2}
\end{align*}
\end{footnotesize}

\section*{References}

\setlength{\bibsep}{0pt plus 0.3ex}
\begin{footnotesize}
\bibliographystyle{elsarticle-harv} 
\biboptions{authoryear}
\bibliography{CAD}
\end{footnotesize}


%\begin{thebibliography}{00}
%\bibitem{}
%
%\end{thebibliography}

\end{document}